\documentclass[a4paper,onecolumn,11pt,accepted=2022-02-14]{quantumarticle}
\pdfoutput=1

\usepackage{hyperref}
\usepackage{doi}
\usepackage[numbers,sort&compress]{natbib}

\usepackage[left=3cm,right=3cm,top=3cm,bottom=3cm]{geometry}
\usepackage{braket}
\usepackage[utf8]{inputenc}  
\usepackage{amsmath, amsthm, amssymb}  
\usepackage{bbm} 
\usepackage{url} 
\usepackage{color}
\usepackage{enumerate}
\usepackage{enumitem}
\usepackage{mathrsfs} 
\usepackage{stmaryrd}
\usepackage{easybmat}
\usepackage{bm}
\usepackage{amsfonts}
\usepackage{hyperref}
\usepackage{cleveref}
\usepackage{mathtools}

\usepackage{thmtools}
\usepackage{thm-restate}

\usepackage{algorithm}
\usepackage{mathtools}
\usepackage{relsize,exscale}
\usepackage[hang,small]{caption}
\usepackage{float}
\usepackage{amsthm}
\usepackage{multicol}
\usepackage{tikz}

\usepackage{array,makecell}
\newcolumntype{R}[1]{>{\raggedleft\arraybackslash }m{#1}}
\newcolumntype{L}[1]{>{\raggedright\arraybackslash }m{#1}}
\newcolumntype{C}[1]{>{\centering\arraybackslash }m{#1}}
\newcolumntype{I}{!{\vrule width 1.5pt}}
\newlength\epaisLigne

\DeclareMathOperator{\Ima}{Im}

\DeclarePairedDelimiter\floor{\lfloor}{\rfloor}

\makeatletter
\def\hlinewd#1{%
  \noalign{\ifnum0=`}\fi\hrule \@height #1 %
  \futurelet\reserved@a\@xhline}
\makeatother

\def\fract#1/#2{\hbox{\leavevmode
  \kern.1em \raise .25ex \hbox{\the\scriptfont0 $#1$}\kern-.1em }\big/
  {\hbox{\kern-.15em \lower .5ex \hbox{\the\scriptfont0 $#2$}} }}

\newtheorem{theorem}{Theorem} 
\newtheorem{lemma}[theorem]{Lemma}

\newtheorem{corol}[theorem]{Corollary}
\newtheorem{defn}[theorem]{Definition}
\newtheorem{corol-defn}[theorem]{Corollary-Definition}
\newtheorem{prop}[theorem]{Proposition}

\usepackage{mathtools}
\usepackage[normalem]{ulem}

\newcommand{\chain}{\mathbf{c}}
\newcommand{\cyc}{\zeta}
\newcommand{\cocyc}{\zeta}
\newcommand{\F}{\mathbb F}
\newcommand{\bC}{\mathbbm{C}}
\newcommand{\bF}{\mathbbm{F}}

\newcommand{\cC}{\mathcal{C}}

\begin{document}

\title{Towards local testability for quantum coding}

\author{Anthony Leverrier}
\affiliation{Inria, France}
\email{anthony.leverrier@inria.fr}
\orcid{0000-0002-6707-1458}
\author{Vivien Londe}
\affiliation{Microsoft, France}
\email{vivien.londe@microsoft.com}

\author{Gilles Z\'emor}
\affiliation{Institut de Math\'ematiques
de Bordeaux, UMR 5251, France}
\email{zemor@math.u-bordeaux.fr}
\orcid{0000-0002-6041-9554}

\maketitle

\begin{abstract}
We introduce the \emph{hemicubic codes}, a family of quantum codes obtained by associating qubits with the $p$-faces of the $n$-cube (for $n>p$) and stabilizer constraints with faces of dimension $(p\pm1)$. The quantum code obtained by identifying antipodal faces of the resulting complex encodes one logical qubit into $N = 2^{n-p-1} \tbinom{n}{p}$ physical qubits and displays local testability with a soundness of $\Omega(1/\log(N))$ beating the current state-of-the-art of $1/\log^{2}(N)$ due to Hastings. 
We exploit this local testability to devise an efficient decoding algorithm that corrects arbitrary errors of size less than the minimum distance, up to polylog factors. 

We then extend this code family by considering the quotient of the $n$-cube by
arbitrary linear classical codes of length $n$. We establish the parameters of
these \emph{generalized hemicubic codes}. Interestingly, if the soundness of the
hemicubic code could be shown to be constant, similarly to the ordinary
$n$-cube, then the generalized hemicubic codes could yield quantum locally testable codes of
length not exceeding an exponential or even polynomial function of the code
dimension.\footnote{An extended abstract of this work appeared at ITCS 2021.}
\end{abstract}

\section{Introduction}

\subsection{Quantum LDPC codes, local testability and robustness of entanglement}

Entanglement is arguably the central concept of quantum theory and despite decades of study, many questions about it remain unsolved today. One particular mystery is the robustness of phases of highly entangled states, such as the ones involved in quantum computation. Given such a state, does it remain entangled in the presence of noise? A closely related question concerns low-energy states of local Hamiltonians: while ground states, \textit{i.e.}, states of minimal energy, are often highly entangled, is it also the case of higher energy states? These questions are related through the concept of quantum error correction: logical information is often encoded in a quantum error correcting code (QECC) in order to be processed during a quantum computation, and the ground space of a local Hamiltonian is nothing but a special case of a QECC called quantum low-density parity-check (LDPC) code.

Physically it indeed makes sense to implement quantum error correction by relying on local interaction, for example by encoding the quantum state in the degenerate ground space of a local Hamiltonian, that is an $N$-qubit operator $H \propto \sum_i \Pi_i$, where each $\Pi_i$ is a projector acting nontrivially on a small number $q$ of qubits (we talk of $q$-local terms). By ``small'', one usually means constant or sometimes logarithmic in $N$. A quantum stabilizer code is a subspace of the space $(\bC^2)^{\otimes N}$ of $N$ qubits defined as the common $+1$ eigenspace of a set $\{S_1, \ldots, S_m\}$ of commuting Pauli operators, that is, the space 
\[ \mathrm{span} \{ |\psi\rangle \in  (\bC^2)^{\otimes N} \: : \: S_i |\psi\rangle = |\psi\rangle, \forall i \in [m] \}.\]
Such a code is said to be \emph{LDPC} if all the generators $S_i$ act nontrivially on at most $q$ qubits for small $q$. 
With this language, a quantum LDPC stabilizer code corresponds to the ground space of the local Hamiltonian $H = \frac{1}{m} \sum_{i=1}^m \Pi_i$, with $\Pi_i = \frac{1}{2}(\mathbbm{1}- S_i)$.

Entanglement can be quantified in many ways, but a relevant definition is to say that a quantum state is highly entangled (or displays \emph{long-range entanglement}) if it cannot be obtained by processing an initial product state via a quantum circuit $U_{\mathrm{circ}}$ of constant depth. By contrast, a quantum state that can be obtained that way, and which is therefore of the form $U_{\mathrm{circ}} \big(\otimes_{i=1}^n |\phi_i\rangle\big)$ for some $|\phi_i\rangle \in \bC^2$, is said to be \emph{trivial}. An important property of trivial states is that one can efficiently compute the value of local observables such as $\Pi_i$ for such states: this is because the operator $U_{\mathrm{circ}}^\dagger \Pi_i U_{\mathrm{circ}}$ remains local (since the circuit has constant depth) and its expectation can therefore be computed efficiently for a product state.  In particular, such a classical description can serve as a witness that a local Hamiltonian admits a trivial state of low energy.
It is well known how to construct $N$-qubit Hamiltonians with highly entangled
ground states, for instance by considering a Hamiltonian associated with a
quantum LDPC code with non-constant minimum distance \cite{BHV06}, but the
question of the existence of local Hamiltonians such that low-energy states are
nontrivial remains poorly understood.

The no low-energy trivial state (NLTS) conjecture asks whether there exists a local Hamiltonian such that all states of small enough (normalized) energy are nontrivial \cite{has13}. More precisely, is there some $H= \frac{1}{m} \sum_{i=1}^m \Pi_i$ as above, such that there exists a constant $\alpha >0$ such that all states $\rho$ satisfying $\mathrm{tr} (\rho H) \leq \alpha$ are nontrivial? What is interesting with the NLTS conjecture is that it is a consequence of the quantum PCP conjecture \cite{AAV13},  and therefore corresponds to a possible milestone on the route towards establishing the quantum PCP conjecture. We note that there are several versions of the quantum PCP conjecture in the literature, corresponding to the quantum generalizations of equivalent versions of the classical PCP theorem, but not known to be equivalent in the quantum case, and that the multiprover version was recently established \cite{NV18}. Here, however, we are concerned with the Hamiltonian version of the quantum PCP which still remains wide open. This conjecture is concerned with the complexity of the \emph{Local Hamiltonian} problem: given a local Hamiltonian as before, two numbers $a < b$ and the promise that the minimum eigenvalue of the Hamiltonian is either less than $a$, or greater than $b$, decide which is the case. The quantum PCP conjecture asserts that this problem is QMA-hard when the gap $b-a$ is constant. This generalizes the PCP theorem that says that the satisfiability problem is NP-hard when the relative gap is constant \cite{din07}. Here, QMA is the class of languages generalizing NP (more precisely generalizing MA), where the witness can be a quantum state and the verifier is allowed to use a quantum computer. Assuming that $\text{NP} \not\subseteq \text{QMA}$, we see that Hamiltonians with trivial states of low energy cannot be used to prove the quantum PCP conjecture since the classical description of such states would be a witness that could be checked efficiently by a classical verifier. In other words, if the quantum PCP conjecture is true, it implies that NLTS holds. The converse statement is unknown.

Eldar and Harrow made progress towards the NLTS conjecture by establishing a
simpler variant, called NLETS \cite{EH17}, by giving an explicit local
Hamiltonian where states close to ground states are shown to be nontrivial. (See also Ref.~\cite{NVY18} for an alternate proof exploiting approximate low-weight check codes.) The
subtlety here is that closeness isn't defined as ``low energy'' as in NLTS, but
by the existence of a low weight operator mapping the state to a ground state.
Viewing the ground space as a quantum LDPC code, \cite{EH17} shows that states
which are $\delta N$-close to the code (for some sufficiently small $\delta>0$) are nontrivial. The NLTS conjecture asks for something stronger: that all states with energy less than a small, constant, fraction of the operator norm of the Hamiltonian are nontrivial. Of course, states close to the codespace have a low (normalized) energy or syndrome weight since each qubit is only involved in a small number of generators, but the converse does not hold in general, and this is what makes the NLTS conjecture difficult to tackle.

One case where the distance to the code is tightly related to the syndrome weight is for \emph{locally testable codes} (LTC): classical locally testable codes are codes for which one can efficiently decide, with high probability, whether a given word belongs to the code or is far from it, where efficiency is quantified in the number of queries to the coordinates of the word. To see the link between the two notions, the idea is to distinguish between codewords and words far from the code by computing a few elements of the syndrome and deciding that the word belongs to the code if all these elements are zero. An LTC is such that any word at constant relative distance from the code will have a constant fraction of unsatisfied checknodes, that is a syndrome of weight linear in the blocklength. The Hadamard code which maps a $k$-bit word $x$ to a string of length $2^k$ corresponding to the evaluations at $x$ of all linear functions provides such an example with the syndrome corresponding to all possible linearity tests between the bits of the word: indeed, any word that satisfies most linearity tests can be shown to be close to the codespace \cite{BLR93}.

While LTCs have been extensively studied in the classical literature \cite{gol05} and provide a crucial ingredient for the proof of the classical PCP theorem, their quantum generalization is relatively new and much less understood. The concept was only recently introduced in a paper by Aharonov and Eldar \cite{AE15} which showed that the classical approaches to local testability seem to fail in the quantum world: for instance, defining a code on a (hyper)graph with too much expansion seems to be a bad idea. 
In any case, if quantum LTC with constant minimum distance existed, they would provide a proof of the NLTS conjecture \cite{EH17}, and this motivates trying to understand whether such codes can exist. 
Let us, however, mention that while classical LTCs are useful in Dinur's combinatorial proof of the classical PCP theorem \cite{din07}, the same doesn't seem to apply in the quantum regime since it is known that directly quantizing Dinur's combinatorial proof of the PCP theorem is bound to fail \cite{BH13, AAV13}.

An additional difficulty in the quantum case is that good quantum LDPC codes are not even known to exist. While taking a random LDPC code yields a code with linear minimum distance with high probability in the classical case, the same statement is not known to hold in the quantum setting. Even restricting our attention to codes only encoding a constant number of logical qubits, it was surprisingly hard until very recently to find families of codes with minimum distance much larger than $\sqrt{N}$: a construction due to Freedman, Meyer and Luo gives a minimum distance $\Theta (N^{1/2} \log^{1/4} N)$ \cite{FML02} while recent constructions based on high-dimensional expanders yield a polylogarithmic improvement \cite{KKL16,EKZ20,KT20}.  (Note that considering subsystem codes \cite{pou05} or approximate codes \cite{CGS05,BO10} is helpful to get a large minimum distance \cite{BFH15,NVY18,BCN19}.)
In 2020, new constructions managed to significantly beat this $\sqrt{N}$ barrier: the fiber bundle codes achieve $N^{3/5}$, up to polylogarithmic factors \cite{HHO20}, and the lifted product codes almost achieve linear minimum distance \cite{PK20}.
For these reasons, while a lot of work on classical LTC is focused on codes with linear minimum distance and aims at minimizing the length of the code, the current goals in the quantum case are much more modest at this point.

A possible formal definition of a quantum LTC was suggested by \cite{EH17}, which we detail now. Recall that the objective is to relate two notions: the distance of a state to the code, and the energy of the state. A quantum code, or equivalently, its associated Hamiltonian, will be locally testable if any word at distance $t$ from the code (or the ground space) has energy $\Omega(t)$ and if this energy can be estimated by accessing only a small number of qubits (this is why we insist on having local terms in the Hamiltonian). 
First, one defines a quantum version of the Hamming distance as follows. 
Consider the code space $\cC \subset (\mathbbm{C}^2)^{\otimes N}$ and define its $t$-fattening $\cC_t$ as the span of states at distance at most $t$ from $\cC$: 
\begin{align*}
\cC_t := \mathrm{Span}\{ (A_1 \otimes \cdots \otimes A_n) |\psi\rangle\: : \: |\psi\rangle \in \cC, |\{i \: : \: A_i \neq \mathbbm{1}\}| \leq t\},
\end{align*}
where the $A_i$ are single-qubit Pauli matrices. States at distance $t$ belong to $\cC_t$, but not to $\cC_{t-1}$, which we formalize by considering the projector $\Pi_{\cC_t}$ onto $\cC_t$ and forming the \emph{distance operator}
\begin{align*}
D_\cC := \sum_t t (\Pi_{\cC_t} -\Pi_{\cC_{t-1}}).
\end{align*}
Informally, the eigenspace of $D_\cC$ with eigenvalue $t$ corresponds to states which are at distance $t$ from the code. 
We now define the averaged normalized Hamiltonian $H_\cC$ associated with the quantum code $\cC$ with $q$-local projections $(\Pi_1, \ldots, \Pi_m)$:
\begin{align*}
H_\cC = \frac{1}{m} \sum_{i=1}^m \Pi_i. 
\end{align*}
The normalization by $m$ ensures that $\|H_{\cC}\| \leq 1$.
With these notations, we say that a $q$-local quantum code $\mathcal{C} \subseteq (\mathbbm{C}^2)^{\otimes n}$ is an $(s,q)$-quantum LTC with soundness $s \in [0,1]$ if\footnote{In a previous version of this manuscript, \url{https://arxiv.org/abs/1911.03069v1}, we were additionally normalizing the Hamiltonian by $q$, leading to a soundness value of $s/q$. We remove this extra factor here, in accordance with the literature in classical and quantum locally testable codes.}
\begin{align}\label{eqn:def}
H_\cC \succeq \frac{s}{N} D_{\mathcal{C}},
\end{align}
where $A \succeq B$ means that the operator $A-B$ is positive semidefinite. In words, condition \eqref{eqn:def} means that any low-energy state is close to the codespace in terms of the quantum Hamming distance, and that simple energy tests allow one to distinguish codewords from states far from the code.  
More precisely, one can distinguish between a codeword (with energy 0) and a state at distance $\delta N$ from the code (therefore with energy $\geq s\delta$) by measuring approximately $1/(s\delta)$ terms of the Hamiltonian.
Ideally, one would want the soundness $s$ and the locality $q$ to be constant, so that accessing a constant number of qubits would suffice to distinguish codewords from states at distance greater than $\delta N$ from the code, for constant $\delta>0$.

Known constructions of quantum LTC are rare. For instance, quantum expander codes almost yield one example of $(s, q)$-quantum LTC with both $s= O(1), q=O(1)$, but with the \emph{major caveat} that Eq.~\eqref{eqn:def} doesn't hold in general, but only on the restriction of the Hilbert space consisting of states $O(\sqrt{N})$-close to the codespace \cite{LTZ15}. In fact, there exist states at distance $\Omega(\sqrt{N})$ violating only a single projection $\Pi_i$. This means in order for Eqn.~\eqref{eqn:def} to hold, one actually needs to take $s = O(1/\sqrt{n})$. Thus, quantum expander codes cannot be used to establish the NLTS conjecture.
By allowing the locality to be logarithmic in the number of qubits instead of constant, that is $q = O(\log N)$, a recent construction of Hastings \cite{has16} yields a quantum LTC with soundness $s = O\left( \frac{1}{\log^2 N} \right)$, without any restriction on the validity of Eq.~\eqref{eqn:def}. The construction is a generalization of the toric code where instead of taking the product of two 1-cycles of length $p$, one rather considers the product of two $d$-cycles of area $p^d$ for the appropriate values of $p = \omega(1)$ and $d=\omega(1)$.

\paragraph{Our results.} In this work, we present a different construction of quantum LTC which shares with Hastings' the property that it is set in a high-dimensional space with $d = \Theta(\log N)$ and therefore a similar locality\footnote{We note that in both our construction and Hastings', each qubit is only involved in a logarithmic number of constraints.} $q = \Theta(\log N)$. Our code, however, achieves a slightly better soundness $r = \Omega\left(\frac{1}{\log N}\right)$, and in fact, we were not able to rule out that the soundness isn't constant, which would be optimal. 
While this hemicube code only encodes a single logical qubit, we can introduce a generalized family of codes with polynomial rate. These codes are
obtained starting with the chain complex associated to the $n$-dimensional
Hamming cube, where we identify faces corresponding to the same coset of a
classical code of length $n$. A CSS quantum code is obtained by placing qubits
on the $p$-faces and stabilizers either on $(p-1)$-faces or $(p+1)$-faces, with
constraints given by the incidence relations between the faces in the cube.
While this construction is arguably quite natural, computing the parameters (dimension and minimum distance) of this code family turned out to be rather subtle, relying in nontrivial arguments from algebraic topology. The parameters of the CSS code resulting from the quotient of the cube by a linear code of parameters $[n,k,d]$ are 
\[ \left\llbracket 2^{n-p-k}\tbinom{n}{p}, \tbinom{p+k-1}{p}, \min\left\{ \tbinom{d}{p}, 2^{n-p-k} \right\} \right\rrbracket \]
when qubits are placed on $p$-faces for $p \leq d-2$.
Whether these codes are also locally testable is left as an open question. In
that case, these would provide the first examples of quantum LTC of exponential
or even polynomial length in the code dimension. Remember indeed that both the hemicubic and Hastings' codes have constant dimension.

\subsection{Construction of the hemicubic code}

We start with the simplest member of our quantum code family, corresponding to the quotient of the $n$-cube by the repetition code. 
It has been known since Kitaev \cite{kit03} that one can associate a quantum CSS code
with any chain complex of binary vector spaces of the form: $C_2  \xrightarrow{\partial_{2}} C_1 \xrightarrow{\partial_1} C_0$, where the boundary operators $\partial_2$ and $\partial_1$ satisfy $\partial_1 \partial_2 = 0$.
One first defines two classical codes $\cC_X = \ker \partial_1$ and $\cC_Z =
(\mathrm{Im}\, \partial_{2})^\perp = \ker \partial_2^T$. These codes satisfy
$\cC_Z^\perp \subseteq \cC_X$ since $\partial_1  \partial_2 = 0$ and the
resulting quantum CSS code is the linear span of $\left\{ \sum_{z \in
\cC_Z^\perp} |x+z\rangle \: : \: x\in \cC_X\right\}$, where $\left\{ |x\rangle\:
: \: x\in \mathbbm{F}_2^N\right\}$ is the canonical basis of
$(\mathbbm{C}^2)^{\otimes N}$ and $N$ is the dimension of the central space
$C_1$ of the chain complex.
One obtains in this way a quantum code of length $N$ and dimension $\mathrm{dim} (\cC_X/\cC_Z^\perp) = \mathrm{dim} (\cC_X) + \mathrm{dim} (\cC_Z) -N$. 
Its minimum distance is given by $d_{\min} = \min (d_X, d_Z)$ with $d_X = \min\{ |w| \: : \: w \in C_X \setminus C_Z^\perp\}$ and $d_Z = \min\{ |w|\: : \: w \in C_Z \setminus C_X^\perp\}$. Here, $|w|$ stands for the Hamming weight of the word $w$.

Our construction relies on the $n$-dimensional hemicube, where a $p$-face is
formed by a pair of antipodal $p$-dimensional faces of the Hamming cube
$\{0,1\}^n$. A $p$-face of the Hamming cube is a string of $n$-elements from
$\{0,1, *\}$ where symbol $*$ appears exactly $p$ times. Let us denote by
$C_p^n$ the $\mathbbm{F}_2$-vector space spanned by $p$-faces of the hemicube.
Boundary $\partial_p$ and coboundary $\delta_p$ operators are obtained by
extending the natural operators for the Hamming cube to the hemicube
\begin{align*}
\partial_p  \, x_1 \ldots x_n & := \bigoplus_{i \, \text{s.t.} x_i=*} x_1 \ldots x_{i-1} 0 x_{i+1} \ldots x_n \oplus x_1 \ldots x_{i-1} 1 x_{i+1} \ldots x_n\\
\delta_p  \, x_1 \ldots x_n & := \bigoplus_{i \, \text{s.t.} x_i \neq *} x_1 \ldots x_{i-1} * x_{i+1} \ldots x_n
\end{align*}
and are further extended to $p$-chains by linearity. We reserve the notation $+$ for the standard notation in $\mathbbm{F}_2$ and use $\oplus$ for summing chains. 
The hemicubic code is then defined as the CSS code obtained from the chain complex
\[C_{p+1}^n \xrightarrow{\partial_{p+1}} C_p^n \xrightarrow{\partial_p} C^n_{p-1}.\] 

Choosing $p = \alpha n$ for $0< \alpha <1$, the resulting code will be LDPC with generators of logarithmic weight since the boundary and coboundary operators act nontrivially on $O(n) = O(\log N)$ coordinates. 
The dimension of the hemicubic code corresponds to that of the homology groups
$H_p^n = \fract{\ker \partial_p}/{\mathrm{Im} \, \partial_{p+1}}$. Since the
hemicube, viewed as a cellular complex, has the same topology as the real projective plane, its homology groups all have the same dimension equal to 1. 
We note that the quantum code obtained here can be described with a completely different approach exploiting Khovanov homology \cite{aud13}.
Obtaining the minimum distance of the code requires more care since one needs to
find lower bounds on the weight of minimal nontrivial cycles and cocycles in the
hemicube. Summarizing, we establish the following result. 
\begin{restatable}{theorem}{hemicubic}\label{thm:main-code}
The hemicubic code is a CSS code with parameters 
\[\left\llbracket N =2^{n-p-1} \tbinom{n}{p} , 1, d_{\min} = \min \left\{\tbinom{n}{p}, 2^{n-p-1} \right\} \right\rrbracket.\]
\end{restatable}

Let $\alpha^* \approx 0.227$ be the unique nonzero solution of $h(\alpha^*) = 1-\alpha^*$ where $h$ is the binary entropy function. Then choosing $p = \left\lfloor \alpha^* n\right\rfloor$ yields a quantum code family with $d_{\min} \geq \frac{\sqrt{N}}{1.62}$ \cite{aud13}.

\subsection{Local testability of the hemicubic code}

We now turn our attention to the local testability of the hemicubic code. This
property results from isoperimetric bounds on the hemicube.

\begin{restatable}{theorem}{LTC}\label{thm:ltc}
The hemicubic code is locally testable with soundness $s = \Omega\left(\frac{1}{\log N}\right)$.
\end{restatable}
This improves over Hastings' construction \cite{has16} obtained by taking the
product of two $n$-spheres and which displays soundness $s =
\Theta\left(\log^{-2} (N)\right)$. It would be interesting to understand whether
the bounds of Theorem \ref{thm:ltc} are tight or not. At the moment, we
believe it might be possible to get rid of the logarithmic factor and obtain a constant soundness for the hemicubic code. This would
then match the soundness of the standard (not projective) Hamming cube, which does not encode any logical qubit
since its associated complex has zero homology.

We say that a $p$-chain $X$ is a \emph{filling} of $Y$ if $\partial X=Y$ and that a $p$-cochain $X$ is a \emph{cofilling} of $Y$ if $\delta X = Y$.
The main tools to establish the soundness of the hemicubic code are upper bounds on the size of fillings (resp. cofillings) for boundaries (resp. coboundaries) in the cube. 
Denoting the Hamming weight of chains and cochains by $\|\;\|$, we have:

\begin{lemma} \label{lem:fill-cofill}
Let $E$ be a $p$-chain of $C_p^n$. Then there exists a $p$-chain $F$ which is a filling of $\partial E$, satisfying $\partial F = \partial E$ such that 
\begin{align*}
\|F\| \leq  \frac{n-p}{2} \|\partial E\|.
\end{align*}
Let $E$ be a $p$-cochain of $C_p^n$. Then there exists a $p$-cochain $F$ which
is a cofilling of $\delta E$, satisfying $\delta F = \delta E$ such that 
\begin{align*}
\|F\| \leq (p+1)  \|\delta E\|.
\end{align*}
\end{lemma}

It is straightforward to translate these results in the language of quantum codes. Let us represent an arbitrary Pauli error of the form $\bigotimes_{i \in E_X, j\in E_Z} X^{i} Z^j$ by a couple $E = (E_X, E_Z)$ where $E_X$ is the support of the $X$-type errors and $E_Z$ is the support of the $Z$-type error. Interpreting $E_X$ as a $p$-chain and $E_Z$ as a $p$-cochain, we see that the syndrome of $E$ is given by the pair $(\partial E_X, \delta E_Z)$.
In order to compute the soundness of the quantum code, one needs to lower bound the ratio:
\begin{align*}
\min_{(E_X,E_Z)} \frac{\|\partial E_X\| + \|\delta E_Z\|}{\|[E_X]\| + \|[E_Z]\|}
\geq \min \left\{\min_{E_X} \frac{\|\partial E_X\| }{\|[E_X]\|}, \min_{E_Z}
\frac{ \|\delta E_Z\|}{ \|[E_Z]\|}  \right\},
\end{align*}
where the minimum is computed over all errors with a nonzero syndrome, i.e., for $p$-chains $E_X$ which are not a $p$-cycle and $p$-cochains $E_Z$ which are not a $p$-cocycle. In these expressions, we denote by $[E]$ the representative of the equivalence class of error $E$, with the smallest weight. Indeed, recall that two errors differing by a element of the stabilizer group are equivalent. 
The fact that one considers $[E]$ instead of $E$ makes the analysis
significantly subtler in the quantum case than in the classical case. A solution
is to work backward (as was also done by Dotterrer in the case of the Hamming cube \cite{dot16}): start with a syndrome and find a small weight error giving rise to this syndrome. 
This is essentially how we establish Lemma \ref{lem:fill-cofill}:
\begin{align*}
\min_{E_X, \partial E_X \neq 0} \frac{\|\partial E_X\| }{\|[E_X]\|} \geq
\frac{2}{n-p}, \quad \min_{E_Z, \delta E_Z \neq 0} \frac{ \|\delta E_Z\|}{
\|[E_Z]\|}  \geq \quad \frac{1}{p+1}. 
\end{align*}
This implies the soundness in Theorem \ref{thm:ltc} since $n-p, p+1 = \Theta(\log N)$. 

While Dotterrer established tight bounds for the size of (co)fillings in the
Hamming cube, we don't know whether the bounds of Lemma \ref{lem:fill-cofill}
are tight. Right now, we lose a logarithmic factor in the case of the hemicube,
but it is not clear that this should be the case. In fact, it is not even excluded
that the hemicube could display a \emph{better} soundness than the standard cube. We expand on these ideas in Section \ref{sec:LTC}.

\subsection{An efficient decoding algorithm for the hemicubic code}

The existence of the small fillings and cofillings promised by the soundness of the code is particularly interesting in the context of decoding since it guarantees the existence of a low-weight error associated to any low-weight syndrome. To turn this into an efficient decoding algorithm, the main idea is to notice that one can efficiently find the required fillings and cofilings and therefore find Pauli errors giving the observed syndrome. While finding the smallest possible fillings or cofillings does not appear to be easy, finding ones satisfying the bounds of Lemma \ref{lem:fill-cofill} can be done efficiently.

We note, however, that the decoding algorithm does not seem to perform so well against random errors of linear weight. In particular,  arguments from percolation theory that would imply that errors tend to only form small clusters and that therefore it is sufficient to correct these errors (similarly to \cite{FGL18} for instance) will likely fail here because of the logarithmic weight of the generators. Indeed, the factor graph of the code has logarithmic degree and there does not exist a constant threshold for the error probability such that below this threshold, errors appear in clusters of size $o(N)$. In addition, and more importantly, our decoding algorithm isn't local in the sense that it explores only the neighborhood of some violated constraints to take a local decision, and for this reason, it is not entirely clear whether the algorithm processes disconnected clusters of errors independently.

\begin{restatable}{theorem}{decoding}\label{thm:decoding-intro}
The hemicubic code comes with an efficient decoding algorithm that corrects adversarial errors of weight $w=O(d_{\min}/\log^2 N)$ with complexity $O(n^4 w)$.
\end{restatable}
The decoding complexity is quasilinear in the error size and the algorithm can be parallelized to run in logarithmic depth. 
Finding a filling (or cofilling) can be done recursively by fixing one of the $n$ coordinates and finding fillings in the projective cube of dimension $n-1$. While the choice of the special coordinate is not immediately obvious if one wants to find the smallest filling, it is nevertheless possible to make a reasonably good choice efficiently by computing upper bounds on the final filling size for each possible choice of coordinate. We establish Theorem \ref{thm:decoding-intro} in Section \ref{sec:decoding}.

\subsection{Generalized hemicubic codes: quotients by arbitrary linear codes}

A key remark is that identifying antipodal $p$-faces of the $n$-cube is equivalent to considering the cosets of the repetition code $\{0^n, 1^n\}$ in the cube complex. It is therefore tempting to generalize this approach by identifying the elements of the cosets of arbitrary linear codes $\cC$ with parameters $[n,k,d]$. 
We form in this way a new complex where two $p$-faces $x$ and $y$ are identified if there exists a codeword $c\in \cC$ such that $x = y +c$. Recall that addition is coordinate-wise here and that $*$ is an absorbing element. 

Deriving the parameters of the quantum CSS code associated to these new
complexes has been surprisingly challenging. In particular it does not seem
particularly obvious that the quantum parameters, especially the minimum
distance, should depend only on the
parameters $[n,k,d]$ of the classical code $\cC$ and not otherwise on its
particular structure: it turns out indeed to be the case however. 
We managed to derive the quantum parameters by exhibiting explicit
representatives of the $\F_2$-homology and cohomology classes, through a double
induction on $p$ and the classical code dimension $k$. We obtain a lower bound
on the minimum homologically nontrivial cycle weight by exhibiting a set of representatives of a
cohomology class all of which must be orthogonal to the cycle, and in particular
intersect it. Since a nontrivial cycle meets this bound it is exact. A similar
method is used to derive the minimum nontrivial cocycle weight and we obtain
the following theorem.

\begin{restatable}{theorem}{param} \label{thm:gen-hemicube}
The quantum code obtained as the quotient of the $n$-cube by a linear code $[n,k,d]$ admits parameters 
\[ \left\llbracket 2^{n-p-k}\tbinom{n}{p}, \tbinom{p+k-1}{p}, \min\left\{ \tbinom{d}{p}, 2^{n-p-k} \right\} \right\rrbracket \]
when qubits are placed on $p$-faces for $p \leq d-2$.
\end{restatable}

An interesting case is $k=2$, which yields a quantum code of exponential length (that is, dimension logarithmic in the code length):
\[ \left\llbracket 2^{n-p-2}\tbinom{n}{p}, p+1, \min\left\{ \tbinom{d}{p}, 2^{n-p-2} \right\} \right\rrbracket. \]

We are only able to prove a lower bound on the soundness of the code (for
$X$-errors) of $\Omega(1/p!)$. However, a much improved soundness would follow
from the conjectured filling and cofilling constants of the original hemicubic
complex: generalized hemicubic codes are therefore candidates for quantum locally
testable codes of growing dimension, of which no examples are presently known.

\subsection{Discussion and open questions}

In this paper, we have introduced a family of quantum code constructions that
live on the quotient of the $n$-dimensional Hamming cube by classical linear
codes. Despite the apparent simplicity of the construction, it does not seem to
have appeared before in the literature. Deriving the parameters of these codes
turned out to be significantly subtler than expected, and quite surprisingly, the parameters of the quantum code only depend on the parameters of the classical code and not any on additional structure. 
The simplest member of our quantum code family, the hemicubic code, basically inherits its local testability from the soundness of the Hamming cube, which was established by Dotterrer. In our view, the fact that our code construction relies so much on the Hamming cube 
may be expected to yield additional advantages, through the import of other
interesting properties from the cube, as well as tools from Boolean analysis.

The most pressing question is to understand whether the generalized hemicubic codes also display local testability. At the moment, we can only establish it for the simplest member of the family, which only encodes a single logical qubit. If we could show that the codes corresponding to the quotient of the Hamming cube by arbitrary linear codes of dimension $k$ remain locally testable, then this would provide the first examples of quantum locally testable codes of exponential (if $k>1$) or polynomial (if $k =\Omega(n)$) length. 
As we discuss in Section \ref{sec:LTC}, improving our bound on the soundness of the one-qubit hemicubic code from $\frac{1}{\log N}$ to constant would already prove that the generalized code with $k=2$ remains locally testable. 
An indication that such an improvement might be possible comes from the 0-qubit code defined on the standard hypercube (without identifying antipodal faces) which indeed displays constant soundness \cite{dot13}. 
More generally, the question of what parameters are achievable for quantum locally testable codes is essentially completely open at the moment.

Another intriguing question is whether the hemicubic code might help towards establishing the NLTS conjecture (albeit with a quasilocal Hamiltonian with terms of logarithmic weight) or more generally whether it is relevant for many-body physics. As mentioned, any quantum LTC with linear minimum distance would yield such a proof \cite{EH17}. The hemicubic code, however, is restricted by a $O(\sqrt{N})$ minimum distance, and the argument of \cite{EH17} doesn't directly apply anymore. This is in particular a line of research followed by Eldar which relies on the hemicubic code and which provides positive partial results \cite{eld19}.
We note that in the physics context of the Local Hamiltonian, it is crucial that
every individual quantum system (say, qubit) is acted upon by a small number of
terms. In this sense, the problem is somewhat more constrained than in the classical local
testability case where one is typically fine if the number of constraints is much larger than the number of qubits. Our quantum codes satisfy this requirement since each qubit is only involved in a logarithmic number of local constraints.

Finally, while classical LTCs have found a number of applications in recent years, notably for constructing PCPs, it is fair to say that not much is presently known about possible applications of quantum
LTCs. At the same time, local testability is a notion that makes perfect sense
in the quantum regime and it seems reasonable to think that quantum LTCs might also find applications. Finding explicit families encoding a non-constant number of qubits is a natural first step.

\subsection*{Outline of the manuscript}  In Section \ref{sec:prel}, we introduce the main notions of algebraic topology needed for the description of our codes and review the notion of local testability both in the classical and the quantum settings.
In Section \ref{sec:code}, we describe the construction of the one-qubit hemicubic code corresponding to the quotient of the $n$-cube by the repetition code and derive its parameters. We consider the general case of quotients by arbitrary linear codes in Section \ref{sec:linear}.
In Section \ref{sec:LTC}, we establish the local testability of the hemicubic code.
Finally, in Section \ref{sec:decoding}, we exploit the local testability of the code to devise an efficient decoding algorithm that runs in quasilinear time, and that can be parallelized to logarithmic depth.

\subsection*{Acknowledgments}
We would like to thank Benjamin Audoux, Alain Couvreur, Omar Fawzi, Antoine Grospellier and Jean-Pierre Tillich for many fruitful discussions on quantum codes. We would like to warmly thank an anonymous reviewer, whose remarks and corrections greatly improved the quality of this article.
AL and VL acknowledge support from the ANR through the QuantERA project QCDA.


\newpage

\section{Preliminaries}

\label{sec:prel}

\subsection{Notions of algebraic topology}
\label{long exact sequence}

We introduce here the notion of chain complex 	as well as the Long Exact Sequence theorem that will be a crucial tool to study the dimension of the quantum codes we will consider in Section \ref{sec:linear}.

Possible references for this section are \cite{Wei95} \S 1.3 p. 10 and \cite{Rot08} Theorem 6.10 p. 333.

A very general definition of a chain complex is the following:

\begin{defn}
A \emph{chain complex} $C_\bullet$ is a sequence of vector spaces $(C_p)_{p
\in \mathbb{Z}}$ and of linear maps $(\partial_p: C_p \rightarrow C_{p-1})_{p \in
\mathbb{Z}}$ called differentials such that the composition of any two
successive differentials is zero: $\partial_{p-1} \partial_{p} = 0$.
\end{defn}

We will limit ourselves to the case when the vector spaces $C_p$ are
finite-dimensional vector
spaces over the binary field $\F_2$  of the form $\F_2^{X_p}$ for $X_p$ some finite set. Elements of $X_p$ are
called {\em $p$-cells} and elements of $C_p$ are called {\em $p$-chains.} Our
chain complexes will also be {\em bounded}, meaning that only a finite number of
spaces $C_p$ will be non-zero.

Elements of $\ker\partial_p$ are called cycles and elements of
$\mathrm{Im}(\partial_{p+1})$ are called boundaries. Every boundary is a cycle but the converse is not necessarily true. Homology groups give information about this phenomenon.

\begin{defn}
The $p$-th \emph{homology group} $H_n(C_\bullet)$ of a chain complex $C_\bullet$
is defined as the quotient $H_p = \ker\partial_p / \mathrm{Im}\partial_{p+1}$. 
\end{defn}

One also defines coboundary operators: $\delta_p: C_p \to C_{p+1}$ \textit{via}
the adjoint (or transpose) map 
$\delta_p = \partial_{p+1}^*$ with the identification of $C_p^*$ with $C_p$ and
of $C_{p+1}^*$ with $C_{p+1}$.
The $p^{\mathrm{th}}$ cohomology group is given by $H^p=\ker \delta_p/\mathrm{Im}\delta_{p-1}$.
A standard fact is that the $p^{\mathrm{th}}$ homology and cohomology groups are isomorphic. Their dimension is called the $p^{\mathrm{th}}$ \emph{Betti number} of the chain complex. 

The following definition and theorem will be used in Section~\ref{sec:linear}.

\begin{defn}
A \emph{chain map} $f$ from the chain complex $C_\bullet$ to the chain complex
$D_\bullet$ is a sequence of morphisms $f_p : C_p \rightarrow D_p$ such that for
all $p \in \mathbb{Z}$, $f_{p-1} \partial_p = \partial_p f_p$. By abuse of notation we use
the same symbol $\partial_p$ to refer to the distinct differentials $\partial_p:
C_p \rightarrow C_{p-1}$ and $\partial_p: D_p \rightarrow D_{p-1}$.
\end{defn}

A chain map sequence $A_\bullet \xrightarrow{f} B_\bullet \xrightarrow{g} C_\bullet$ is called \emph{exact} if for all $p \in \mathbb{Z}$, the sequence $A_p \xrightarrow{f_p} B_p \xrightarrow{g_p} C_p$ is exact, i.e., $\mathrm{Im} f_{p}=\ker g_p$. An exact sequence of the form $0 \rightarrow A_\bullet \xrightarrow{f} B_\bullet \xrightarrow{g} C_\bullet \rightarrow 0$ is called a \emph{short exact sequence}, and other exact sequences are traditionally called \emph{long}.

\begin{theorem}[Long Exact Sequence] \label{thm long}
A short exact sequence $0 \rightarrow A_\bullet \rightarrow B_\bullet \rightarrow C_\bullet \rightarrow 0$ of chain complexes induces the following long exact sequence of homology groups:
$$\ldots \rightarrow H_{p}(A_\bullet) \rightarrow H_{p}(B_\bullet) \rightarrow H_{p}(C_\bullet) \rightarrow H_{p-1}(A_\bullet) \rightarrow H_{p-1}(B_\bullet) \rightarrow \ldots$$
\end{theorem}

We refer to \cite{Wei95} or \cite{Rot08} for a proof. However we will make more explicit the induced morphism $H_{p}(A_\bullet) \rightarrow H_{p}(B_\bullet)$ (or equivalently $H_{p}(B_\bullet) \rightarrow H_{p}(C_\bullet)$) and the connecting morphism $H_{p}(C_\bullet) \rightarrow H_{p-1}(A_\bullet)$. 

The homology group morphism $H_{p}(A_\bullet) \rightarrow H_{p}(B_\bullet)$ is
induced by $f_p: A_p \rightarrow B_p$. It is well defined because $f_p$ takes cycles to cycles and boundaries to boundaries. To avoid confusion we will sometimes denote the chain group morphism by $f_{\mathrm{chain},p}$ and the homology group morphism by $f_{\mathrm{hom},p}$. \\

The connecting morphism $H_{p}(C_\bullet) \rightarrow H_{p-1}(A_\bullet)$ takes more work to construct. 

Let $[c_p]$ be a class in $H_{p}(C_\bullet)$ represented by the element $c_p$ of
$C_p$. There exists $b_p \in B_p$ such that $g_p(b_p) = c_p$. Now,
$g_{p-1}(\partial_p(b_p)) = \partial_p(g_p(b_p)) = \partial_p(c_p) = 0$ because $c_p$ is a cycle.
Therefore there exists $a_{p-1} \in A_{p-1}$ such that $f_{p-1}(a_{p-1}) =
\partial_p(b_p)$. The connecting morphism is defined by sending $[c_p]$ to $[a_{p-1}]$. \\
We leave it to the reader to prove that $a_{p-1}$ is a cycle, that its class
$[a_{p-1}]$ in $H_{p-1}(A_\bullet)$ doesn't depend on the representative
$c_{p-1}$ chosen for $[c_{p-1}]$ and that the connecting map actually is a
morphism. To avoid confusion we will sometimes denote the chain group
differential by $\partial_{\mathrm{chain},p}$ and the connecting homology group
morphism by $\partial_{\mathrm{hom},p}$.

In the present work, we will form a chain complex associated with the $n$-dimensional Hamming cube or with quotients of this cube by linear codes, and the space $\mathcal{C}_p$ will be the $\mathbbm{F}_2$-space spanned by $p$-faces of the resulting cube.

\subsection{CSS codes}

A quantum code encoding $k$ logical qubits into $N$ physical qubits is a subspace of $(\mathbbm{C}^2)^{\otimes N}$ of dimension $2^k$. A simple way to define such a subspace is \textit{via} a stabilizer group, that is an abelian group of $N$-qubit operators (tensor products of single-Pauli operators $X = \left(\begin{smallmatrix} 0 & 1 \\ 1&0\end{smallmatrix}\right)$, $Z = \left(\begin{smallmatrix} 1 & 0 \\ 0&-1\end{smallmatrix}\right)$, $Y=ZX$ and $\mathbbm{1}$ with an overall phase $\pm1$ or $\pm i$) that does not contain $-\mathbbm{1}$. A \emph{stabilizer code} is then defined as the eigenspace of the stabilizer with eigenvalue $+1$ \cite{got97}. A stabilizer code of dimension $k$ can be described by a set of $N-k$ independent generators of its stabilizer group. Note, however, that in the context of locally testable codes, it will be natural to consider larger sets of generators, to allow for some extra-redundancy. The \emph{minimum distance} $d_{\min}$ of a quantum code is the minimum weight of a nontrivial logical operator, that is an operator that commutes with all the elements of the stabilizer, but does not belong to the stabilizer. A quantum code of length $N$ encoding $k$ qubits with minimum distance $d_{\min}$ is denoted $\llbracket N, k, d_{\min} \rrbracket$.

CSS codes are a special case of stabilizer codes where the generators are either products of Pauli-$X$ and $I$, or products of Pauli-$Z$ and $I$ \cite{CS96,ste96,ste96b}. These families are easier to study because the commutation relations required to make the stabilizer abelian simply need to be checked between $X$-type and $Z$-type generators. In particular, such quantum codes can be described by a pair of classical codes.
\begin{defn}[CSS code]
A quantum CSS code with parameters $\llbracket N, k, d_{\min} \rrbracket$ is a
pair of classical codes (\textit{i.e.}, $\F_2$-vector spaces) $\cC_X, \cC_Z \subseteq \mathbbm{F}_2^N$ such that
$\cC_X^\perp \subseteq \cC_Z$, or equivalently $\cC_Z^\perp \subseteq \cC_X$. It
corresponds to the linear span of $\left\{ \sum_{z \in \cC_Z^\perp} |x+z\rangle
\: : \: x\in \cC_X\right\}$, where $\left\{ |x\rangle\: : \: x\in \mathbbm{F}_2^N\right\}$ is the canonical basis of $(\mathbbm{C}^2)^{\otimes N}$.
\end{defn}

The dimension of the code is given by $\dim(\cC_X/\cC_Z^\perp) =
\dim\cC_X + \dim\cC_Z - N$ and its minimum distance is
given by $d_{\min} = \min (d_X, d_Z)$ with $d_X = \min\{ |w| \: : \: w \in \cC_X
\setminus \cC_Z^\perp\}$ and $d_Z = \min\{ |w|\: : \: w \in \cC_Z \setminus \cC_X^\perp\}$. Here, $|w|$ stands for the Hamming weight of the word $w$.

Quantum LDPC codes are stabilizer codes coming with a list of low-weight generators. For instance, a CSS code $CSS(\cC_X, \cC_Z)$ is said to be LDPC if both $\cC_X$ and $\cC_Z$ are given with sparse parity-check matrices. Here sparse means that the weight of each row (or equivalently the weight of the corresponding generator) is constant or logarithmic in the length $N$. These codes are particularly interesting because low-weight constraints are more realistic in terms of implementation. Moreover, one can exploit this sparsity to design efficient decoders. 

A chain complex $\cC_{p+1}^n \xrightarrow{\partial_{p+1}} \cC_p^n
\xrightarrow{\partial_p} \cC^n_{p-1}$ gives rise to a CSS code by considering the classical codes $\cC_X= \ker \partial_p$ and $\cC_Z = (\mathrm{Im} \, \partial_{p+1})^\perp$. Indeed, the condition $\mathrm{Im}\, \partial_{p+1} \subseteq \ker \partial_p$ immediately implies that $\cC_Z^\perp \subseteq \cC_X$. In that case, the dimension of the quantum code is simply the $p^{\mathrm{th}}$ Betti number of the chain complex. 

It will also be convenient to introduce the parity-check matrices $H_X$ and $H_Z$ of the classical codes $\cC_X$ and $\cC_Z$, so that $\cC_X = \ker H_X$ and $\cC_Z = \ker H_Z$. With the correspondance between chain complexes and CSS codes outlined above, we get
\[ H_X = \partial_p, \quad H_Z = \partial_{p+1}^T=\delta_p,\]
and they satisfy $H_X \cdot H_Z^T = 0$.

An \emph{error pattern} is defined as a couple $(e_X, e_Z)$ where $e_X$ and
$e_Z$ are both binary vectors. The \emph{syndrome} associated to this error
consists in fact of a couple of syndromes $\sigma_X = H_X e_X^T$ and $\sigma_Z =
H_Z e_Z^T$. A decoder for the code $CSS(\cC_X, \cC_Z)$ is given the pair
$(\sigma_X, \sigma_Z)$ and decoding succeeds if it outputs a couple of error
candidates of the form $(e_X + f_X, e_Z + f_Z)$ with $f_X \in \mathrm{Im} H_Z^T$
and $f_Z \in \mathrm{Im} H_X^T$. The presence of $(f_X, f_Z)$ is a crucial
difference with the classical setting and results from the fact that the
associated operators act trivially on the codespace. It will be useful to
keep in mind that the boundary and coboundary operators $\partial_p$ and
$\delta_p$ are nothing but the syndrome functions for the associated quantum
code.

\subsection{Local testability}

Let us first quickly review the notion of local testability in the classical setting. In this case, the distance $\mathrm{dist}(w,\cC)$ of a word $w \in \F_2^n$ of a word to a classical code $\cC$ is defined as expected by:
\[ \mathrm{dist}(w,\cC) = \min_{c \in \cC} |w+c|,\]
where $|x|$ is the Hamming weight of $x$, that is the number of non-zero bits of $x$.

\begin{defn}
A code $\cC \subseteq \bF_2^N$ with parity-check matrix $H \in \F_2^{m\times N}$ is said to be a $(q,s)$-\emph{locally testable code} with \emph{soundness} $s>0$ if the rows of $H$ have weight at most $q$ and if
\begin{align}\label{eqn:cltc}
\frac{1}{m} |Hw| \geq s \frac{\mathrm{dist}(w, \cC)}{N}
\end{align}
holds for any word $w \in \bF_2^N$. Here, $Hw$ is the syndrome of the word $w$ and $\mathrm{dist}(w, \cC)$ is the distance from $w$ to $\cC$, that is, the minimal Hamming distance between $w$ and a codeword $c\in \cC$.
\end{defn}

This definition gives rise to a simple test to distinguish between a codeword and a word at distance at least $\delta N$ from the code: one simply picks $1/(s\delta)$ rows of the parity-check matrix uniformly at random and measures the associated $1/(s\delta)$ bits. If the word $w$ is $\delta N$ away from the code, then Eq.~\eqref{eqn:cltc} implies that $\frac{|Hw|}{m} \geq s\delta$ and therefore testing $O(\frac{1}{s\delta})$ random constraints will be sufficient to detect it with high probability. 
The quantity $1/(s\delta)$ is therefore referred to as \emph{query complexity} of the code: this is the number of bits from $w$ that should be queried to decide whether $w$ is in the code or far from it. 

We see that the defining property of an LTC, and more specifically of its
parity-check matrix, is that the weight of the syndrome of a word (akin to its
``energy'') is lower-bounded by a function of its distance to the code. In particular, we want to avoid errors of large weight with a small syndrome. 
Many constructions of classical LTC are known, for instance the Hadamard code and the long code. Classically, one important open question concerns the existence of short LTCs which display linear minimum distance, constant soundness and constant rate.

The study of quantum locally testable codes (qLTC) was initiated by Aharonov and Eldar with the motivation that such objects could prove useful in order to attack the quantum PCP conjecture \cite{AE15}. While defining local testability for general quantum codes appears rather involved, the situation is much nicer for stabilizer codes. The definition is then analogous to the classical case, with the difference that the functions ``energy'': $w \mapsto \frac{1}{m} |Hw|$ and ``distance'': $w \mapsto \mathrm{dist}(w, \cC)$ need be replaced by Hermitian operators. 
The quantum observable corresponding to the energy is the Hamiltonian operator. Let $\cC$ be a stabilizer code with a set of $m$ $q$-local generators $(S_1, \ldots, S_m)$ of the stabilizer group, meaning that 
\[ \cC = \{ |\psi\rangle \in (\mathbbm{C}^2)^{\otimes n} \: : \: S_i |\psi\rangle = |\psi\rangle, \forall i \in [m]\}.\]
One first forms $m$ projectors $\Pi_i = \frac{1}{2} (\mathbbm{1} - S_i)$ so that the codespace becomes the 0 eigenspace of $\sum_i \Pi_i$. Note that the generators $S_i$ of the stabilizer group are products of Pauli operators and therefore admit a spectrum $\mathrm{spec} (S_i) = \{-1, +1\}$.

The (normalized) Hamiltonian $H_\cC$ associated with the code is defined as
\begin{align*}
H_\cC = \frac{1}{m} \sum_{i=1}^m \Pi_i.
\end{align*}
This is the straightforward generalization of the notion of syndrome weight to the quantum case. 
Defining the distance to the code requires more care, however, in the quantum setting.
We follow the approach from Ref.~\cite{EH17}.
The idea is to define a set of $N+1$ subspaces 
\[ \cC_0 := \cC \quad \subseteq \quad \cC_1 \quad \ldots \quad \cC_{N-1} \quad \subseteq \quad \cC_N = \F_2^N \]
such that $\cC_t$ corresponds to the space of states at distance $t$ from the code. 
More precisely, $\cC_t$ is the $t$-fattening of $\cC$:
\begin{align*}
\cC_t := \mathrm{Span}\{ (A_1 \otimes \cdots \otimes A_n) |\psi\rangle\: : \: |\psi\rangle \in \cC, \#\{i \: : \: A_i \neq \mathbbm{1}\} \leq t\}.
\end{align*}
Let $\Pi_{\cC_t}$ be the projector onto $\cC_t$, so that $(\Pi_{\cC_t} -\Pi_{\cC_{t-1}})$ is the projector onto states at distance $t$ but not at distance $t-1$ from $\cC$. 
A state $|\psi\rangle  (\bC^2)^{\otimes N}$ is a distance at least $t$ from the code if 
\[ \langle \psi |\Pi_{\cC_{t-1} }|\psi\rangle = 0, \]
which we denote by $ \mathrm{dist}(|\psi\rangle,\cC) \geq t$.
We finally define the operator $D_\cC$ that ``measures'' the distance to the code as
\begin{align*}
D_\cC := \sum_t t (\Pi_{\cC_t} -\Pi_{\cC_{t-1}}).
\end{align*}
In particular, a state is at distance exactly $t$ from the code $\cC$ if it is an eigenstate of the operator $D_\cC$ with eigenvalue $t$.

With these notations, we are ready to define the notion of locally testable code in the quantum case.

\begin{defn}[\cite{EH17}]
A quantum stabilizer code $\cC \subseteq (\bC^2)^{\otimes N}$ is a $(q,s)$-LTC with $q$-local projections $\Pi_1, \ldots, \Pi_m$ if the following operator inequality holds
\begin{align}\label{eqn:LTC}
\frac{1}{m} \sum_{i=1}^m \Pi_i \succeq \frac{s}{N} D_{\cC}.
\end{align}
\end{defn}

If $\cC$ is a $(q,s)$-LTC, the following lower bound holds:
\begin{align*}
 \min_{|\psi\rangle, \mathrm{dist}(|\psi\rangle,\cC) \geq \delta N}  \langle \psi |H_{\cC} |\psi\rangle \geq s \delta,
\end{align*}
where we optimize over all states at distance at least $\delta N$ from the code.

Similarly to the classical case, the \emph{query complexity} of the quantum LTC is given by $\frac{1}{s \delta}$ since it is sufficient to measure this number of qubits to distinguish between a codeword and a state at distance at least $\delta N$ from the code.

In order to prove that a CSS code is LTC, it is sufficient to show that both classical codes $\cC_X, \cC_Z$ are LTC (see \cite{AE15}).

\begin{lemma}
Let $\cC_X$ and $\cC_Z$ be classical $(q,s)$-locally testable codes with $m_X$ and $m_Z$ parity-checks, respectively. 
Then the quantum code $CSS(\cC_X, \cC_Z)$ is $(q,s')$-locally testable with $s' = s \min\big( \frac{m_X}{m_X+m_Z}, \frac{m_Z}{m_X+m_Z}\big)$.
\end{lemma}

\begin{proof}
The idea is to consider a common eigenbasis of $H_\cC$ and $D_\cC$ and to prove that Eqn.~\eqref{eqn:LTC} holds for this basis. 
One starts with $\mathrm{dim}(\cC_X/\cC_Z^\perp)$ elements of the form $\sum_{z \in \cC_Z^\perp} |x+z\rangle$ for $x\in \cC_X$. These are codewords and belong to $\cC_0 = \cC$. Then one completes this family by applying to these states bit-flip errors and phase-flip errors characterized by binary vectors $e_X, e_Z \in \F_2^N$, to get states $|\psi_{e_X, e_Z}\rangle$ of the form $\sum_{z \in \cC_Z^\perp} (-1)^{e_Z\cdot(x+e_X +z)^T} |x+e_X +z\rangle$. 
Alternatively, one can obtain this state by applying the Pauli operator 
\[X^{e_X} Z^{e_Z} := (\bigotimes_{i \in \mathrm{supp} (e_X)} X_i)(\bigotimes_{i \in \mathrm{supp} (e_Z)} Z_i),\]
where $X_i$ and $Z_i$ are the Pauli-$X$ and Pauli-$Z$ operators $\left( \begin{smallmatrix} 0 & 1\\ 1 & 0\end{smallmatrix} \right)$ and $\left( \begin{smallmatrix} 1 & 0\\ 0 & -1\end{smallmatrix} \right)$ applied to the $i^{\mathrm{th}}$ qubit. 

Such a state $|\psi_{e_X, e_Z}\rangle$ belongs to the eigenspace of $H_\cC$ corresponding to energy $\frac{1}{m}( |H_X e_X| + |H_Z e_Z|)$, and more precisely to the subspace of states with syndrome $(\sigma_X, \sigma_Z) =(H_X e_X, H_Z e_Z)$. Note that the states $|\psi_{e_X, e_Z}\rangle$ define a complete basis of eigenstates of $H_{\cC}$. The state also has full support on some $\cC_t \setminus \cC_{t-1}$ for some $t$ that depends on $(e_X, e_Z)$. An easy upper bound for $t$ is given by
\[ t \leq \min_{c_X\in \cC_X} |e_X + c_X| + \min_{c_Z\in \cC_Z} |e_Z + c_Z|, \]
that is, the distance from the state to the code is upper-bounded by the sum of the distance from $e_X$ to $\cC_x$ and the distance from $e_Z$ to $\cC_Z$. This is because $|\psi_{e_X,e_Z}\rangle$  and $|\psi_{e_X+c_X,e_Z+c_Z}\rangle = X^{c_X} Z^{c_Z} |\psi_{e_X,e_Z}\rangle$ are at the same distance for the quantum code.

Assume now that both $\cC_X$ and $\cC_Z$ are $(q,s)$-locally testable, then by definition, it holds that
\[ \frac{1}{m_X} |H_X e_X| \geq \frac{s}{N} \min_{c_X\in \cC_X} |e_X + c_X|, \quad \frac{1}{m_Z} |H_Z e_Z| \geq \frac{s}{N} \min_{c_Z\in \cC_Z} |e_Z + c_Z|,\]
where $m_X$ and $m_Z$ are the number of parity-checks for $\cC_X$ and $\cC_Z$, respectively, and $m=m_X+m_Z$. 
The energy of $|\psi_{e_X, e_Z}\rangle$ satisfies:
\begin{align*}
\langle  \psi_{e_X, e_Z}|H_\cC |\psi_{e_X, e_Z}\rangle&= \frac{1}{m} (|H_X e_X| +|H_Z e_Z| ) \\
&  \geq  \frac{ s m_X}{N (m_X+m_Z)} \min_{c_X\in \cC_X} |e_X  +c_Z|+ \frac{ s m_Z}{N(m_X+m_Z)} \min_{c_Z\in \cC_Z} |e_Z + c_Z|\\
& \geq \frac{s \min(m_X,m_Z)}{N (m_X+m_Z)}\langle \psi_{e_X, e_Z} | D_\cC |\psi_{e_X, e_Z}\rangle,
\end{align*}
which was to be proven.
\end{proof}

\section{The hemicubic code}

\label{sec:code}

In this section, we will consider the simplest member of our code family, corresponding to the quotient of the $n$-cube by the repetition code. Quotients by arbitrary linear codes will be studied in detail in Section \ref{sec:linear}.

\subsection{The construction}
\label{subsec:code}

Let $Q^n=\{0,1\}^n$ be the Hamming cube for $n \geq 2$. For $p \in [n]$, a
$p$-face (or $p$-cell) $x$ of $Q^n$ in an element of $\{0,1,*\}^n$ with exactly $p$ indeterminates, denoted with stars, $|x|_* = p$.
Let $Q_p^n$ be the set of the $p$-faces in the cube. Its cardinality is $|Q_p^n| = 2^{n-p} \tbinom{n}{p}$. We also define the space of $p$-chains of the cube to be the $\mathbbm{F}_2$-vector space spanned by the $p$-faces of the cube. 
We note that the symbols 0 and 1 can appear either as scalars of $\mathbbm{F}_2$
or as letters of $p$-faces. In the text, there shouldn't be any ambiguity
between these different uses. We also note that we alternatively use the set
notation or chain notation to describe a chain: for instance $\{00**, 0*0*\}$
and $00** \oplus 0*0*$ represent the same object. We reserve the notation $+$
for the standard addition in $\mathbbm{F}_2$ or $\F_2^n$  and use $\oplus$ for summing chains.

Boundary $\partial_p$ and coboundary $\delta_p$ operators can be defined in the usual way for $p$-faces of the cube:
\begin{align*}
\partial_p  \, x_1 \ldots x_n & := \bigoplus_{i \, \text{s.t.}\, x_i=*} x_1 \ldots x_{i-1} 0 x_{i+1} \ldots x_n \oplus x_1 \ldots x_{i-1} 1 x_{i+1} \ldots x_n\\
\delta_p  \, x_1 \ldots x_n & := \bigoplus_{i \, \text{s.t.} x_i \neq *} x_1 \ldots x_{i-1} * x_{i+1} \ldots x_n,
\end{align*}
and extended to arbitrary $p$-chains by linearity.

The \emph{hemicube} is formed by considering equivalence classes (of $p$-faces)
of $Q^n$ for the equivalence relation between $p$-faces defined as 
\begin{align*}
x \sim y \iff y = x +11 \ldots 1 =: \bar{x},
\end{align*}
where addition between 0-faces and $p$-faces is defined pointwise with the convention $0 + * = 1+ * = *$. 

Let $C_p^n$ be the $\mathbbm{F}_2$-vector space spanned by $p$-faces of the
hemicube, that is with the identification $\sim$. Formally, a $p$-face of
the hemicube (Hamming cube with the identification $\sim$) is a pair of elements $\{x, \bar{x}\}$, but we will often abuse notation and denote it by either $x$ or $\bar{x}$ when there is no ambiguity. 
An element of $C_p^n$ is called a $p$-chain of the hemicube. The boundary operator $\partial_p$ defined for $p$-faces can be extended to $C_p^n$
because it commutes with the $+ \, 1 \ldots 1$ relation that defines elements of $C_p^n$:
\begin{align*}
\partial_p (\bar{x}) = \overline{\partial_p(x)}.
\end{align*}
One obtains the following chain complex:
\begin{align*}
C_{n-1}^n\xrightarrow{\partial_{n-1}} C_{n-2}^n \xrightarrow{} \cdots \xrightarrow{} C_{p+1}^n \xrightarrow{\partial_{p+1}} C_p^n \xrightarrow{\partial_p} C^n_{p-1} \xrightarrow{} \cdots C_0^n \xrightarrow{\partial_0} 0,
\end{align*}
where we will often write $\partial$ instead of $\partial_p$ for simplicity.

One can then form the homology groups $H_p^n = \ker \partial_p/\mathrm{Im}
\, \partial_{p+1}$. Since the hemicube, viewed as a cellular complex, has the same topology as the real projective plane, its homology groups have the same dimension.
\begin{theorem}\label{thm:hom}
\begin{align*}
H_p^n \cong \mathbbm{F}_2, \quad \dim H_p^n = 1, \quad \forall p \in [n-1].
\end{align*}
\end{theorem}

\begin{proof}
Topologically, the Hamming cube is equivalent to the $n$-sphere, and $Q^n\!/\!\!
\sim$ is equivalent to the projective space $\mathbbm{RP}^n$. The claim follows
from the $\mathbbm{F}_2$ homology of $\mathbbm{RP}^n$: $H_p(\mathbbm{RP}^n, \mathbbm{F}_2) = \mathbbm{F}_2$.
\end{proof}

We will give a more general proof of this fact in Section \ref{sec:linear} where we address the general case of quotients of the cube by arbitrary linear codes.

We denote the associated CSS code by ${\mathcal Q}_p^n$ and the following theorem gives its parameters:
\begin{theorem}\label{thm:main-code2}
The quantum CSS code $(\cC_X,\cC_Z)$ associated with the chain complex
$C_{p+1}^n \xrightarrow{\partial_{p+1}} C_p^n \xrightarrow{\partial_p}
C^n_{p-1}$ with $\cC_X= \ker \partial_p$ and $\cC_Z = (\mathrm{Im}\, \partial_{p+1})^\perp$ has parameters $[[N, 1, d_{\min}]]$, with 
\begin{align*}
N &= 2^{n-p-1} \tbinom{n}{p},\\
d_{\min} &= \min \left\{\tbinom{n}{p}, 2^{n-p-1} \right\}.
\end{align*}
\end{theorem}

We note that the quantum code obtained here can be described with a completely different approach exploiting Khovanov homology \cite{aud13}.

\begin{proof}
The length of the code is simply the length of $\cC_X$, that is the dimension of
$C_p^n$, which is half the cardinality of $Q_p^n$. The number of $p$-cells in $Q_p^n$ is the number of choices of $p$ positions for the stars in a string of length $n$, times $2^{n-p}$ possible bit-strings for the remaining coordinates. This yields the result. 

Computing the minimum distance is less easy: in Lemma \ref{lem:logical}, we give
explicit representatives of the $p^{\mathrm{th}}$ homology and cohomology
groups, yielding in particular $\tbinom{n}{p}$ disjoint cohomologically
nontrivial cocycles.
We argue in addition that any homologically nontrivial cycle
necessarily meets all these nontrivial cocycles and is therefore of weight at
least $\tbinom{n}{p}$.
A somewhat similar argument shows that the weight of nontrivial cocycle is
at least $2^{n-p-1}$, thus completing the proof.
\end{proof}

Let us start by constructing $p$-cycles, that we will denote by ${n\brack p}$ recursively.
We describe the cycles by giving representative in the original cube complex.
We start by defining ${n\brack 1}$ and ${n \brack n}$:
\begin{align*}
{n\brack 1} :=& \bigoplus_{ \ell \in [n-1]} 0^{n-\ell-1} * 1^\ell \nonumber \\
=& \quad \, \, \, 000 \ldots \, \, 0* \\
& \,\oplus\,  00 \ldots 0*1 \\
& \,\oplus\, 0 \ldots 0*11 \\ 
& \,\oplus\, \ldots \\
& \,\oplus\, *1 \ldots 1 , \\ 
{n \brack n} :=& \, *\ldots *,
\end{align*}
where ${n \brack 1}$ is a 1-cycle while ${ n\brack n}$ is only defined formally
(since there are no $n$-cells in the $n$-dimensional hemicube).
For $p \leq n$, we further define
\begin{align}\label{eqn:np}
{n\brack p} := {n-1 \brack p} \frac{1 + (-1)^{p+1}}{2}  \oplus {n-1 \brack p-1} *, 
\end{align}
where $ \frac{1 + (-1)^{p+1}}{2}$ is either 0 or 1 depending on the parity of $p$,
$S \alpha$ is the chain obtained by concatenating $\alpha \in \{0,1, \!*\}$ at the end of each
face of the chain $S$.

\begin{lemma}\label{lem:cycl}
The $p$-chain ${n \brack p}$ is a representative of a $p$-cycle of weight
$\tbinom{n}{p}$ in the hemicube. 
\end{lemma}

For a set $S \in Q_p^n$, we define $\overline{S}$ to be the set obtained by
exchanging 0 and 1 in every element of $S$. In particular, the projections of
$S$ and $\overline{S}$ are
identical in the hemicube.

\begin{proof}
Let us prove by induction that 
\begin{align*}
\partial {n \brack p} = {n \brack p-1} \oplus \overline{n \brack p-1}.
\end{align*}
The base cases $\partial {n \brack 1}$ and $\partial{n \brack n}$ are easily verified:
\begin{align*}
\partial {n \brack 1 } &= \bigoplus_{\ell \in [n-1]} 0^{n - \ell} 1^{\ell} \,\oplus \, \bigoplus_{\ell \in [n-1]} 0^{n - \ell -1} 1^{\ell+1} \\
&= 00 \ldots 0 \, \oplus \, 11\ldots 1 = { n \brack 0} \, \oplus \, \overline{n \brack 0},
\end{align*}
where we formally defined ${n \brack 0} := 00 \ldots 0$. And
\begin{align*}
\partial {n \brack n} &= \bigoplus_{\ell=0}^{n-1} \left( *^\ell 0 *^{n-\ell-1} \, \oplus \, *^\ell 1 *^{n-\ell-1} \right) = {n \brack n-1}\, \oplus \, \overline{n \brack n-1}.
\end{align*}
Let us establish the induction step:
\begin{align*}
\partial {n \brack p} &= \partial {n-1 \brack p} \frac{1 + (-1)^{p+1}}{2}  \, \oplus \, {n-1 \brack p-1} 0 \, \oplus \, {n-1 \brack p-1} 1  \, \oplus \, \partial {n-1 \brack p-1} *\\
&= \left( {n-1 \brack p-1} \frac{1 + (-1)^{p+1}}{2} \, \oplus \, \overline{n-1 \brack p-1} \frac{1 + (-1)^{p+1}}{2} \right)\, \oplus \, \left( {n-1 \brack p-1} 0 \, \oplus \, {n-1 \brack p-1} 1 \right) \\
& \quad \, \oplus \, \left({n-1 \brack p-2} * \, \oplus \, \overline{n-1 \brack p-2} * \right) \\
&=\left( {n-1 \brack p-1} \frac{1 + (-1)^{p-1}}{2} \, \oplus \, \overline{n-1 \brack p-1} \frac{1 + (-1)^{p+1}}{2} \right)  \, \oplus \, \left( {n-1 \brack p-2} * \, \oplus \, \overline{n-1 \brack p-2} * \right) \\
&=\left( {n-1 \brack p-1} \frac{1 + (-1)^{p-1}}{2} \, \oplus \, {n-1 \brack p-2} * \right) \, \oplus \, \overline{\left({n-1 \brack p-1} \frac{1 + (-1)^{p-1}}{2} \, \oplus \, {n-1 \brack p-2} * \right)}\\
&={n \brack p-1} \, \oplus \, \overline{n \brack p-1}.
\end{align*}

The identification of ${n \brack p-1}$ and $\overline{n \brack p-1}$ shows that
that $\partial {n \brack p}=0$ in the hemicube, and therefore that ${n\brack p}$
represents a cycle.
In order to compute its cardinality, we also proceed by induction: the base cases are immediate and the induction step follows from the formula for the binomial coefficient: $\tbinom{n}{p} = \tbinom{n}{p-1} + \tbinom{n-1}{p-1}$.
\end{proof}

\begin{lemma}\label{lem:cycle2}
The $p$-chain ${n \brack p}$ can alternatively be described as the sum of $p$-faces represented by strings of length $n$, with all $\tbinom{n}{p}$ possible positions for the $p$ stars, and with the remaining indices filled with 0 and 1 as follows: a 0 when there is an even number of stars on its left and a 1 when there is an odd number of stars on its left. 
\end{lemma}
For instance, we have

\begin{align*}
{4 \brack 2} &= *\!*00 \\
&\oplus\, *1\!*\!0 \\
&\oplus\, *11* \\
&\oplus\, 0\!*\!*0 \\
&\oplus\, 0\!*\!1* \\
&\oplus\, 00\!*\!*.
\end{align*}

\begin{proof}
The proof is again by induction on $n$ and $p$: the base cases ${n \brack 1}$ and ${n \brack n}$ are immediate.
For the induction step, it is sufficient to see that Eq.~\eqref{eqn:np} holds for the alternate description of Lem.~\ref{lem:cycle2}: observe that for every $p$-face of ${n \brack p}$, if the last coordinate is not a star, it takes value $0$ if $p$ is even and $1$ if $p$ is odd. Since $p$ is the number of stars on the left of the last coordinate,  the description of Lem.~\ref{lem:cycle2} holds for ${n \brack p}$ assuming that it holds for ${n \brack p-1}$ and ${n-1 \brack p-1}$.
\end{proof}
\begin{figure}
\begin{center}
\begin{tikzpicture}[scale=1.25]
  \tikzstyle{every node}=[circle, inner sep=0.05cm, draw]

  \node (0) at (0,0) {$0$};
  \node (a) at (45:2cm) {$a$};
  \node (b) at (0,2) {$b$};
  \node (c) at (135:2cm) {$c$};
  \node (d) at (-2,0) {$d$};
  \path (b) ++(45:2cm) node[label=right:{$a+b$}] (a+b) {$\phantom{a}$};
  \path (b) ++(135:2) node (b+c) {$\phantom{a}$};
  \path (b+c) ++(45:2) node[label=above right:{$d+e$}] (a+b+c) {$\phantom{a}$};
  \path (b+c) ++(180:2) node (b+c+d) {$\phantom{a}$};
  \path (a+b+c) ++(180:2) node (a+b+c+d) {$e$};
  \path (c) ++(180:2) node (c+d) {$\phantom{a}$};
  \path (b+c+d) ++(210:2) node (b+c+d+e) {$a$};
  \path (a+b+c+d) ++(210:2) node (a+b+c+d+e) {$0$};
  \path (b+c+d+e) ++(270:2) node[label=left:{$a+b$}] (c+d+e) {$\phantom{a}$};
  \path (c+d+e) ++(-45:2) node[label=below left:{$d+e$}] (d+e) {$\phantom{a}$};
  \path (d+e) ++(0:2) node (e) {$e$};

  \path (a) ++(180:2) node (a+d) {$\phantom{a}$};
  \path (d) ++(90:2) node (b+d) {$\phantom{a}$};
  \path (b+d) ++(45:2) node (a+b+d) {$\phantom{a}$};

  \path (a+b+c+d) ++(270:2) node (b+e) {$\phantom{a}$};
  \path (c) ++ (45:2) node (a+c) {$\phantom{a}$};

  \draw[line width =0.05cm, color=blue] (0) -- (a);
  \draw[line width =0.05cm] (0) -- (b);
  \draw[line width =0.05cm] (0) -- (c);
  \draw[line width =0.05cm] (0) -- (d);
  \draw[line width =0.05cm, color=blue] (b) -- (a+b) ;
  \draw[line width =0.05cm, color=red] (a) -- (a+b) -- (a+b+c) -- (a+b+c+d) -- (a+b+c+d+e);
  \draw[line width =0.05cm] (b) -- (b+c);
  \draw[line width =0.05cm] (c) -- (b+c);
  \draw[line width =0.05cm, color=blue] (b+c) -- (a+b+c) ;
  \draw[line width =0.05cm] (b+c) -- (b+c+d);
  \draw[line width =0.05cm, color=blue] (b+c+d) -- (a+b+c+d) ;
  \draw[line width =0.05cm] (c) -- (c+d) -- (b+c+d);
  \draw[line width =0.05cm] (d) -- (c+d);
  \draw[line width =0.05cm, color=blue] (b+c+d+e) -- (a+b+c+d+e) ;
  \draw[line width =0.05cm] (b+c+d) -- (b+c+d+e);
  \draw[line width =0.05cm, color=red] (b+c+d+e) -- (c+d+e) -- (d+e) -- (e) --(0);
  \draw[line width =0.05cm] (c+d+e) -- (c+d);
  \draw[line width =0.05cm] (d+e) -- (d);

  \draw[style=dashed, line width =0.05cm, color=blue] (d) -- (a+d);

  \draw[style=dashed, line width =0.05cm, color=blue] (b+d) --
  (a+b+d);
  \draw[style=dashed] (d) -- (b+d);
  \draw[style=dashed] (b) -- (b+d);
  \draw[style=dashed] (a) -- (a+d) -- (a+b+d) -- (a+b);
  \draw[style=dashed, line width =0.05cm, color=blue] (c+d) -- (b+e);
  \draw[style=dashed] (b+e) -- (a+b+c+d);
  \draw[style=dashed, line width =0.05cm, color=blue] (c) -- (a+c);
  \draw[style=dashed] (b+e) -- (a+c) -- (a+b+c);
  \draw[style=dashed] (a) -- (a+c);
  \draw[style=dashed] (b+d) -- (b+c+d);
  \draw[style=dashed] (a+b+d) -- (a+b+c+d);
\end{tikzpicture}
\end{center}
\caption{A minimal nontrivial $2$-cycle for $n=5$. The hemicube is obtained from the Cayley graph 
over $\F_2^5$ with
generators $a,b,c,d,e$ by the identification $a+b+c+d+e=0$. The 16 vertices of
the graph are represented, together with the 10 faces that make up the
$2$-cycle. The two red paths are identified after twisting, $2$-faces with dashed edges are not in the cycle.}
\label{fig:2cycle}
\end{figure}
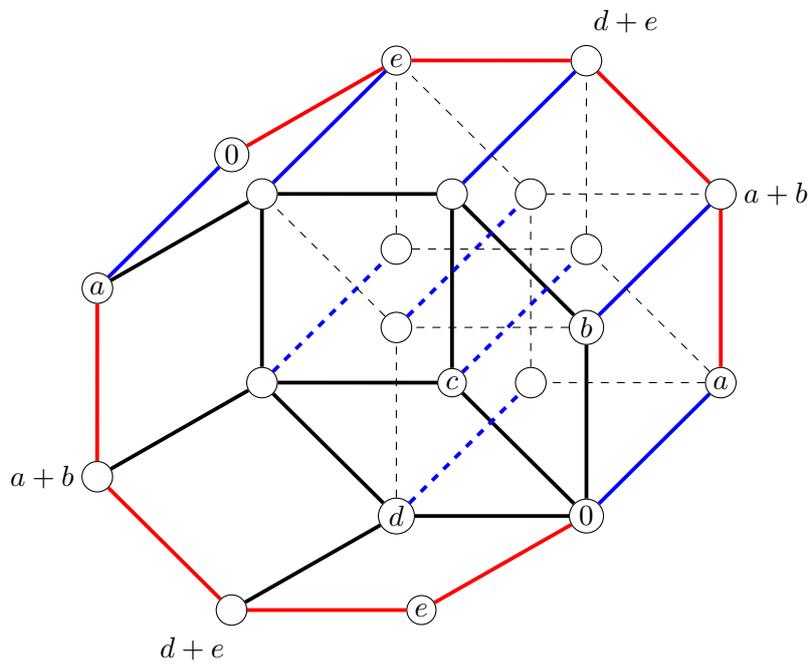
For $n=5$ and $p=2$, a $2$-cycle is represented on Figure~\ref{fig:2cycle}.

\medskip

Let us turn our attention to cocycles. We will show that a representative of a
$p$-cocycle in the hemicube is

\begin{align*}
S_p^n := \bigoplus_{a_i \in \{0,1\}, \, 0\leq  i \leq n-p-1} *\ldots *0a_1 \ldots a_{n-p-1}.
\end{align*}

In words,
 $S_p^n$ represents all the $p$-faces of the hemicube parallel to $*\dots*0\ldots0$. By parallel we mean that their $p$ stars are at the same positions.

\begin{lemma}\label{lem:cocycle}
The set $S_p^n$ represents a $p$-cocycle in the hemicube, that is, $\delta_p S_p^n = 0$.
\end{lemma}

\begin{proof}
The coboundary of $S_p^n$ is the $(p+1)$-chain:
\begin{align*}
\delta S_p^n &=\bigoplus_{i=1}^{n-p-1} \bigoplus_{a_1, \ldots, a_{n-p-1} \in \{0,1\}^{n-p-2}} *\ldots*0a_1 \ldots a_{i-1}*a_{i+1}\ldots a_{n-p-1}.
\end{align*}
Observe that the same term $*\ldots*0a_1 \ldots a_{i-1}*a_{i+1}\ldots a_{n-p-1}$ appears twice, once from $*\ldots*0a_1 \ldots a_{i-1}0a_{i+1}\ldots a_{n-p-1}$ and a second time from $*\ldots*0a_1 \ldots a_{i-1}1a_{i+1}\ldots a_{n-p-1}$. The coboundary of $S_p^n$ is therefore null, which proves the claim. 
\end{proof}

\begin{lemma}
\label{lem:logical}
The cycle ${n\brack p}$ and the cocycle $S_p^n$ are homologically and
cohomologically nontrivial respectively. Furthermore they are of minimum weight
among nontrivial cycles and cocycles.
\end{lemma}

\begin{proof}
We first notice that ${n\brack p}$ and $S_p^n$ meet in exactly one $p$-cell of the hemicube, namely $\{** \ldots * 00\ldots0, ** \ldots * 11\ldots1\}$. This means in particular that ${n\brack p}$ and $S_p^n$ 
are non-orthogonal vectors of $C_p^n$. Both are therefore nontrivial since any
cycle is orthogonal to all trivial cocycles (coboundaries) and any cocycle is
orthogonal to all trivial cycles (boundaries).

We make the further remark that if a nontrivial cycle were orthogonal to a given
nontrivial cocycle, then {\em all} nontrivial cycles, and hence {\em all}
cycles must be also be orthogonal to this given cocycle (because the homology of the hemicube has
dimension $1$). But the orthogonal space of the cycle space is the space of
trivial cocycles, a contradiction. Therefore any nontrivial cycle/nontrivial
cocycle pair must consist of non-orthogonal vectors.

To establish that ${n\brack p}$ is a nontrivial cycle of minimum weight,
we find $\tbinom{n}{p}$ nontrivial cocycles with disjoint supports. 
These are obtained similarly to $S_p^n$ by placing the $p$ stars in all
possible $\tbinom{n}{p}$ positions, instead of the first $p$ coordinates.
Since a nontrivial cycle must be non-orthogonal to, and therefore
intersect, all of them, this proves that its
weight is at least $\tbinom{n}{p}$.

The reasoning is almost similar for the minimum nontrivial cocycle weight (and can be seen as an application of
Theorem 2.8 in \cite{AC15}). For $x \in \{0,1\}^n$, define $T_x {n \brack p}$ to
be the cycle that is the translate of ${n \brack p}$ by the vector $x$. 
This means that $T_x {n \brack p}$ is obtained from  ${n \brack p}$ by replacing
every $p$-face $y$ in ${n \brack p}$ by $y+x$.
Note that any $p$-face $z$ belongs to exactly $2^{p}$ translates 
$T_x {n \brack p}$ (all such $x$ may differ only in the star coordinates of
$z$). All translates $T_x {n \brack p}$ are nontrivial cycles since they are not orthogonal to the nontrivial cocycle $S_p^n$. Indeed they intersect at exactly one $p$-face since $T_x {n \brack p}$ contains one $p$-face for each positions of the $p$ stars and $S_p^n$ contains all the $p$-faces that have the $p$ stars in the first $p$ positions.
Therefore any nontrivial cocycle must be non-orthogonal to, and in particular intersect, all
of them. Since any of its $p$-faces can belongs to exactly $2^p$ translates $T_x
{n \brack p}$ and since the total number of distinct translates $T_x {n \brack p}$ 
in the hemicube complex is $2^{n-1}$, because translating by $x$ or
$\overline{x}$ is equivalent, we get that the weight of any nontrivial cocycle is
at least $2^{n-p-1}$. Since this is the weight of $S_p^n$, the bound is tight.
\end{proof}

We immediately get the value of the minimum distance. 
\begin{corol}
The minimum distance of the quantum code corresponding to $p$-chains in the
$n$-dimensional hemicube is 
\begin{align*}
d_{\min} = \min \left\{\tbinom{n}{p}, 2^{n-p-1}\right\}.
\end{align*}
\end{corol}

Let $h(x) = -x \log x -(1-x)\log (1-x)$ be the binary entropy function where here and throughout, the logarithm is taken in base 2.
Let us define $\alpha^* \approx 0.227$ to be the unique nonzero solution of $h(\alpha^*) = 1-\alpha^*$.
Then choosing $p = \left\lfloor \alpha^* n\right\rfloor$ yields a quantum code where both minimum distances are approximately equal and satisfy:
\begin{align*}
d_{\min} = \Theta(\sqrt{N}). 
\end{align*}
See \cite{aud13} for a detailed analysis of the constant in this equation. Specifically, Proposition 5.5 of Ref.~\cite{aud13} shows that one can extract a subfamily of codes with $d_{\min} \geq \frac{\sqrt{N}}{1.62}$.

Let us finally mention that the number $m$ of constraints in the hemicubic code is of the same order as the number $N$ of qubits since we have
\begin{align*}
N &= 2^{n-p-1}\tbinom{n}{p}\\
m &= 2^{n-p}\tbinom{n}{p+1} + 2^{n-p-2}\tbinom{n}{p-1} = \left(\frac{2(n-p)}{p+1} + \frac{p}{2(n-p+1)} \right) N\\
&= O(N).
\end{align*}
We note that this property is welcome in the physics context of the Local Hamiltonian since the relevant Hamiltonians in Nature should be such that every particle is only acted upon by a reasonably small number of local terms, but that it is much easier to obtain local testability when allowing the number of constraints to be larger than the number of (qu)bits. In fact, such redundancy among constraints is even required to get local testability in the classical case \cite{BGK10}.

\section{Quotient by general linear codes}

\label{sec:linear}

The construction of the quantum code associated with the hemicube can be readily
generalized by realizing that identifying antipodal points in the hypercube
$Q^n$ is in fact equivalent to taking the quotient by a repetition code of
length $n$. It becomes then natural to consider quotients by more general classical linear codes, provided that their minimum distance is large enough (it should be larger than $p+2$). In particular, the quotient of the cube by a classical code of dimension $k$ will yield a quantum code of dimension $\tbinom{k+p-1}{p}$.

\subsection{Dimension of a hemicubic code: algebraic proof}
\label{subsec:hemi}

The $n$-hemicube is the quotient of the cube by the antipodal map. We have already stated that since the $n$-hemicube has the topology of the projective $n$-space, the quantum code obtained from identifying qubits with $p$-faces of the $n$-hemicube has as many logical qubits as the rank of $H_p(\mathbbm{RP}^n, \mathbbm{F}_2)$: it has one logical qubit for $1 \leq p \leq n-2$. In this section, we give a more algebraic proof of this result using the long exact sequence of \S \ref{long exact sequence}. This proof generalizes to other quotients of the $n$-cube. \\

There is a natural projection map from $p$-faces of the cube to $p$-faces of the hemicube defined by sending a face $f_p$ of the $n$-cube to its equivalence class $\floor{f_p}$ in the hemicube:
\[
\forall f_p \in F_p(Q^n), \pi_p(f_p) = \floor{f_p}.\]
For example for $n=3$ and $p=0$, $\pi_p(001) = \pi_p(110) = \floor{001}$. It is straightforward to extend $\pi_p$ by linearity to $\pi_{chain, p}$ from $C_p (Q^n)$ to $C_p (Q^n/\cC_r)$. For example, $\pi_{chain,p}(000 \oplus 001 \oplus 010 \oplus 111) = \floor{001} \oplus \floor{010}$.

Something confusing in our construction is that on the one hand we take quotients of the cube by classical codes and on the other hand homology classes classes are defined as the quotient of a coarse-grained equivalence class (cycles) by a fine-grained equivalence class (boundaries). To minimize confusion between different quotients, we denote classical codes equivalence classes with the floor symbol: $\floor{f_p}$ and homology classes with brackets: $h = [c]$ where the chain $c$ is a sum of $p$-faces: $c = \floor{f_p} \oplus \floor{g_p}$ for instance. \\

We also define an injective map from the hemicube to the cube. This map is directly defined at the level of chains:
\[
\forall \, c = \bigoplus \floor{f_p} \in C_p (Q^n/\cC_r), i_{chain,p}(\bigoplus \floor{f_p}) = \bigoplus f_p \oplus (f_p + 1\ldots 1).\]
We denote by $f_p + 1\ldots 1$ the translate of $f_p$ by the non-zero codeword of the repetition code: $1\ldots 1$. For example,  $i_{chain,p}(\floor{000} \oplus \floor{001} \oplus \floor{010}) = 000 \oplus 111 \oplus 001 \oplus 110 \oplus 010 \oplus 101$. \\

\begin{lemma}
Recall from subsection \ref{long exact sequence} that $C_\bullet(Q^n)$ and $C_\bullet(Q^n/ \cC_r )$ denote the chain complexes of the $n$-cube and the $n$-hemicube. $i_{complex}$ and $\pi_{complex}$ are respectively the collection of maps $i_{chain, p}$ and $\pi_{chain, p}$.  The sequence of chain complexes
\begin{align*}
0 \rightarrow C_\bullet (Q^n/ \cC_r ) \xrightarrow{i_{complex}}
C_\bullet (Q^n) \xrightarrow {\pi_{ccomplex}} C_\bullet (Q^n/ \cC_r ) \rightarrow 0
\end{align*}
is a short exact sequence of chain complexes.
\end{lemma}

\begin{proof}
\begin{itemize}
\item $i_{chain,p}$ commutes with the boundary operator $\partial$.
\item $\pi_{chain,p}$ commutes with the boundray operator $\partial$.
\item $i_{chain,p}$ is injective.
\item $\pi_{chain,p}$ is surjective.
\item $\Ima i_{chain,p} = \ker \pi_{chain,p}$
\end{itemize}
\end{proof}

Theorem \ref{thm long} associates a long exact sequence of homology groups to this short exact sequence of chain complexes:
\begin{align*}
\ldots \rightarrow H_{p}(Q^n/ \cC_r) \rightarrow H_{p}(Q^n) \rightarrow
H_{p}(Q^n/ \cC_r) \rightarrow H_{p-1}(Q^n/ \cC_r) \rightarrow H_{p-1}(Q^n) \rightarrow \ldots \; .
\end{align*}

Since $H_{p}(Q^n) = 0$ for $1 \leq p \leq n$, the exact sequence breaks into small pieces:
\begin{align*}
\forall p \in  \{1, \ldots, n-1\} \, , \, 0 \rightarrow H_{p}(Q^n/\cC_r)
\rightarrow H_{p-1}(Q^n/ \cC_r) \rightarrow 0.
\end{align*}

\noindent
$H_{0}(Q^n/ \cC_r )$ has dimension 1 since the hemicube is path-connected. By immediate induction, $H_{p}(Q^n/ \cC_r )$ has dimension 1 for $1 \leq p \leq  n-1$. \\

Therefore the quantum code constructed by identifying qubits with $p$-faces of the hemicube
for $1 \leq p \leq  n-1$ has dimension $1$. This completes the algebraic proof of Theorem \ref{thm:hom}. \\

In this subsection, we have deduced the homology groups of the hemicube from the homology groups of the cube. In the next sections, we extend by induction this reasoning to more general quotients.

\subsection{Long exact sequences for a generalized hemicube}

To define a generalized hemicubic code we take the quotient of the cube
$Q^n$ by a classical code $\cC$ with parameters $[n, k, d]$ with $d \geq p+2$, thus creating the
quotient polytope $Q^n / \cC$ where faces of the cube are identified when they
are indexed by $n$-tuples that differ by a codeword of~$\cC$. The quantum
code is then associated to the polytope in the usual way
(qubits correspond to $p$-faces of the quotient.)

\begin{lemma}
\label{long_seq}
Let $\cC_k$ be an $[n, k, d_k]$ classical code and $\cC_{k+1}$ be an $[n, k+1, d_{k+1}]$ classical code such that $\cC_{k+1}$ contains $\cC_k$. \\
For $p \leq d_{k+1}$ and $p-1 \geq 0$, the following long sequence of homology groups is exact:
\begin{align*}
\ldots  \xrightarrow{\partial_{hom}} &H_{p}(Q^n/ \cC_{k+1})
\xrightarrow{i_{hom,p}} H_{p}(Q^n/\cC_k) \xrightarrow{\pi_{hom,p}} H_{p}(Q^n/ \cC_{k+1}) \xrightarrow{\partial_{hom}} \\
&H_{p-1}(Q^n/ \cC_{k+1}) \xrightarrow{i_{hom,p-1}} H_{p-1}(Q^n/\cC_k)
\xrightarrow{\pi_{hom,p-1}} H_{p-1}(Q^n/ \cC_{k+1}) \xrightarrow{\partial_{hom}} \ldots 
\end{align*}
For $p \geq 0$ and $p+1 \leq d_{k+1}$, the following long sequence of cohomology groups is exact:
\begin{align*}
\ldots  \xrightarrow{\delta_{cohom}} &H^{p}(Q^n/ \cC_{k+1})
\xrightarrow{i_{cohom,p}} H^{p}(Q^n/ \cC_k) \xrightarrow{\pi_{cohom,p}}
H^{p}(Q^n/ \cC_{k+1}) \xrightarrow{\delta_{cohom}} \\
&H^{p+1}(Q^n/ \cC_{k+1}) \xrightarrow{i_{cohom,p+1}} H^{p+1}(Q^n/\cC_k)
\xrightarrow{\pi_{cohom,p+1}} H^{p+1}(Q^n/ \cC_{k+1}) \xrightarrow{\delta_{cohom}} \ldots 
\end{align*}
\end{lemma}

\begin{proof}
For $p \geq 0$, we can construct
the following short exact sequence of $p$-chain vector spaces: 
\begin{equation}
\label{eq:ipi}
0 \rightarrow C_p (Q^n/ \cC_{k+1}) \xrightarrow{i_{chain,p}} C_p (Q^n/ \cC_{k})
\xrightarrow {\pi_{chain,p}} C_p (Q^n/ \cC_{k+1}) \rightarrow 0
\end{equation}

\noindent
where, like in subsection \ref{subsec:hemi}, $\pi_{chain,p}$ is the linear extension at the level of chains of the projection of $p$-faces from $Q^n/ \cC_{k}$ to
 $Q^n/ \cC_{k+1}$  and $i_{chain,p}$ lifts each $p$-face of a $p$-chain of $C_p (Q^n/\cC_{k+1})$ to the sum of the two corresponding $p$-faces in $C_p (Q^n/\cC_{k})$. The subscripts $chain$, $complex$, $hom$ and $cohom$ indicate that we are considering the corresponding variants of $i$ and $\pi$.

Like in subsection \ref{subsec:hemi}, the sequences \eqref{eq:ipi} for $p \geq 0$ define a short exact sequence of chain
complexes since $i_{chain,p}$ and $\pi_{chain,p}$ commute with the boundary
operator $\partial$: 
\begin{equation}
\notag
0 \rightarrow C_\bullet (Q^n/ \cC_{k+1}) \xrightarrow{i_{complex}} C_\bullet
(Q^n/ \cC_{k}) \xrightarrow {\pi_{complex}} C_\bullet (Q^n/ \cC_{k+1}) \rightarrow 0.
\end{equation}

\noindent
The associated long exact sequence of homology groups gives the first part of the result. \\

The derivation of the long exact sequence of cohomology groups is formally identical. Chains and cochains are canonically identified: they both can be considered as subsets of the set of faces. We define $i_{cochain,p}$ to be equal to $i_{chain,p}$ and $\pi_{cochain,p}$ to be equal to $\pi_{chain,p}$. Since $i_{cochain,p}$ and $\pi_{cochain,p}$ commute with the coboundary operator $\delta$ (the transpose of the boundary operator $\partial$), we have the following short exact sequence of cochain complexes: 
\begin{equation}
\notag
0 \rightarrow C^\bullet (Q^n/ \cC_{k+1}) \xrightarrow{i_{cocomplex}} C^\bullet
(Q^n/ \cC_{k}) \xrightarrow {\pi_{cocomplex}} C^\bullet (Q^n/ \cC_{k+1}) \rightarrow 0.
\end{equation}

\noindent
The associated long exact sequence of cohomology groups gives the second part of the result.
\end{proof}

We highlight the fact that even though $i_{chain, p}$ is the same map as $i_{cochain, p}$, $i_{hom,p}$ can be (and actually is as we will show later) very different from $i_{cohom,p}$. Similarly $\pi_{hom,p}$ is very different from $\pi_{cohom,p}$.

\subsection{Cohomology basis and short exact sequence in cohomology for a
generalized hemicube}
\label{cohomology basis}

The long exact sequence of cohomology groups is actually easier to manipulate
than its homology counterpart. In this subsection we will define a basis of the cohomology group $H^{p}(Q^n/ \cC_k)$ with a prefered representative for each element of the basis. We will see the consequences of this basis for the long exact sequence of cohomology groups.

It will prove useful to denote a cohomology class $cohom$ by one of its representative, \textit{i. e.} by a cocycle $cocyc$ which is not a boundary and which belongs to $cohom$. In this case we have that $cohom = [cocyc]$. The additional knowledge of a preferred representative for each element of the cohomology basis will be crucial.

\begin{defn}
The $p$-direction of a $p$-face $f_p \in F_p(Q^n/\cC_k)$ is the subset of coordinates where $f_p$ has a star: $\{ i \in \{1, \ldots  , n\} \, | \, f_p(i) = *\}$.
\end{defn}

\begin{defn}[Canonical cocycle]
For a $p$-direction $I \subset \{1, \ldots, n\}$ with $|I|=p$, we call the {\em
canonical cocycle of $p$-direction $I$}
the sum of all the $p$-faces in $F_p(Q^n/\cC_k)$ having this $p$-direction~$I$. We denote it by $\cocyc^{I,p,k}$. It is straightforward to verify that it is indeed a cocycle.
\end{defn}

For example for the cube with $n = 3$,  $p = 1$ and the $1$-direction $I = \{3\}$, the canonical cocycle is $\cocyc^{I,1,0} = (00*) \oplus (01*) \oplus (10*) \oplus (11*)$. For the hemicube with the same parameters, it is $\cocyc^{I,1,0} = \floor{00*} \oplus \floor{01*}$. Note that in the hemicube $\floor{00*}$ is the same $1$-face as $\floor{11*}$ and we could write $\cocyc^{I,1,0} = \floor{11*} \oplus \floor{01*}$ as well.
\begin{theorem}
The cohomology group $H^{p}(Q^n/ \cC_{k})$ has a basis such that each basis element is represented by a canonical cocycle.
\end{theorem}

\begin{proof}
We establish the claim by induction over $(p+k)$. The base case was proved in Lemmas \ref{lem:cocycle}
and~\ref{lem:logical}.

Let $(c_1, \ldots, c_{k-1})$ be a basis of $\cC_{k-1}$ and $c_k$ be such that $(c_1, \ldots, c_{k})$ is a basis of $\cC_k$. We consider a fixed position $j \in \text{Supp}(c_k)$. $\text{Supp}(c_k)$ is the set of positions such that $c_k$ has a $1$ at this position. We can assume without loss of generality that for $1 \leq i \leq k-1$, $j \notin \text{Supp}(c_i)$ (just add $c_k$ to $c_i$ if needed). 

By the induction hypothesis, the cohomology group $H^{p}(Q^n/ \cC_{k-1})$ has a
basis such that each basis element is represented by a canonical cocycle. Since
$\pi_{cochain}$ applied to a canonical cocycle gives the empty cochain
$\varnothing$, $\pi_{cohom}$ is zero on $H^{p}(Q^n/ \cC_{k-1})$. For the same
reason $\pi_{cohom}$ is zero on $H^{p-1}(Q^n/ \cC_{k-1})$ and the long exact sequence in cohomology breaks into the following short exact sequences:
\begin{equation}
\notag
0 \rightarrow H^{p-1}(Q^n/ \cC_{k}) \xrightarrow{\delta_{cohom}} H^{p}(Q^n/
\cC_{k}) \xrightarrow{i_{cohom}}  H^{p}(Q^n/ \cC_{k-1}) \rightarrow 0.
\end{equation}

We will use the above short exact sequence and apply the induction hypothesis to
the cohomology groups
$H^{p}(Q^n/ \cC_{k-1})$ and $H^{p-1}(Q^n/ \cC_{k})$:

Let $I \subset \{1, \ldots, n\}$ with $|I|=p$ be a $p$-direction such that
$[\cocyc^{I,p,k-1}]$ is an element of the basis of $H^{p}(Q^n/ \cC_{k-1})$. Since
$i_{cochain}(\cocyc^{I,p,k}) = \cocyc^{I,p,k-1}$, $i_{cohom}([\cocyc^{I,p,k}]) =
[\cocyc^{I,p,k-1}]$. Therefore the basis of cohomology classes of $H^{p}(Q^n/
\cC_{k-1})$ represented by canonical cocycles has a free family of preimages by
$i_{cohom}$ represented by canonical cocycles of $H^{p}(Q^n/ \cC_{k})$.

Let $I \subset \{1, \ldots, n\}$ with $|I|=p-1$ be a $(p-1)$-direction such that
$[\cocyc^{I,p-1,k}]$ is an element of the basis of $H^{p-1}(Q^n/ \cC_{k})$. $j
\notin I$ because $\forall x \in \cC_{k-1}, j \notin \text{Supp}(x)$. Also
because $\forall x \in \cC_{k-1}, j \notin \text{Supp}(x)$, it makes sense to
say that the $j^{th}$ coordinate of a $p$-face of $\cocyc^{I,p-1,k}$ is zero or
one. Keeping only the faces of $\cocyc^{I,p-1,k}$ whose $j^{th}$ coordinate is zero gives a preimage of $\cocyc^{I,p-1,k}$ by $\pi_{cochain}$. 
Applying $\delta_{cochain}$ to this preimage gives $i_{cochain}(\cocyc^{I \cup
\{j\},p,k})$. Since $\delta_{cohom}$ corresponds to $i_{cochain}^{-1} \circ
\delta_{cochain} \circ \pi_{cochain}^{-1}$ at the level of cochains, we obtain
that $\delta_{cohom}([\cocyc^{I,p-1,k}]) = [\cocyc^{I \cup \{j\},p,k}]$. Therefore
the basis of cohomology classes of $H^{p-1}(Q^n/ \cC_{k})$ represented by
canonical cocycles is sent by $\delta_{cohom}$ to a free family of
cohomologically classes represented by canonical cocycles of $H^{p}(Q^n/ \cC_{k})$.

The exactness of the short exact sequence implies that the concatenation of
these two free families forms a basis of $H^{p}(Q^n/ \cC_{k})$.
\end{proof}

As side products, we obtain that $\pi_{cochain,p,k} = 0$ and that the long exact sequences in cohomology break into pieces of small exact sequences:

\begin{equation}
\notag
0 \leftarrow H^{p}(Q^n/ \cC_{k-1}) \xleftarrow{i_{cohom}} H^{p}(Q^n/ \cC_{k})
\xleftarrow{\delta_{cohom}}  H^{p-1}(Q^n/ \cC_{k}) \leftarrow 0.
\end{equation}

We wrote the above short exact sequence in cohomology from right to left to prepare its adjunction property with its homology counterpart.

\subsection{Adjunction and short exact sequence in homology for a generalized
hemicube}

The following ``quasi-equations'' depicted with $\approx$ summarise how the connecting homology and cohomology morphisms are constructed from applications at the level of chains and cochains:

\begin{align}
\label{conecting_morphism_1}
\partial_{hom} &\approx i_{chain}^{-1} \circ \partial_{chain} \circ \pi_{chain}^{-1}, \\
\label{conecting_morphism_2}
\delta_{cohom} &\approx i_{cochain}^{-1} \circ \delta_{cochain} \circ \pi_{cochain}^{-1}.
\end{align}

\noindent
On the right hand side of $\approx$ are (co)chain morphisms and preimages of chain morphisms. On the left hand side of $\approx$ are (co)homology group morphisms. $\approx$ means that if we consider a (co)chain representing a (co)homology class, any preimage or image by the right hand side (co)chain morphisms yields a representative of the image by the left hand side (co)homology morphism. This is true by construction of the connecting (co)homology morphisms $\partial_{hom}$ and $\delta_{cohom}$ as it was described in subsection \ref{long exact sequence}. 

\begin{lemma}
\label{adjoints}
$\pi_{chain}$ and $i_{cochain}$ are adjoint with respect to $\langle \_ \, , \_ \, \rangle$, the standard bilinear form over $\mathbb{F}_2$ (element-wise xor). Since chains and cochains are canonically identified, $\pi_{cochain}$ and $i_{chain}$ are also adjoint.
\end{lemma}

\begin{proof}
By linearity it is sufficient to prove it for chains made of a single $p$-face. Two such chains intersect with respect to the standard bilinear form if and only if they are equal. Let $f_{p,k-1}$ be a $p$-face of $Q^n/ \cC_{k-1}$ and $f_{p,k}$ be a $p$-face of $Q^n/ \cC_{k}$.
\begin{align*}
&\langle \pi_{chain}(f_{p,k-1}), f_{p,k} \rangle = 1 \\
\Leftrightarrow \quad &\pi_{chain}(f_{p,k-1}) = f_{p,k} \\
\Leftrightarrow \quad &i_{cochain} \circ \pi_{chain}(f_{p,k-1})  = i_{cochain} (f_{p,k}) \\
\Leftrightarrow \quad &f_{p,k-1} \oplus (f_{p,k-1} + c_k)  = i_{cochain} (f_{p,k}) \\
\Leftrightarrow \quad &\langle f_{p,k-1}, i_{cochain}(f_{p,k}) \rangle = 1
\end{align*}
The inverse direction in the last equivalence comes from the fact that $i_{cochain}(f_{p,k})$ is always the sum of a $p$-face of $Q^n/ \cC_{k}$ and its translation by $c_k$.
\end{proof}

\begin{lemma}
The long exact sequences in homology break into pieces of small exact sequences:
\begin{equation}
\notag
0 \rightarrow H_{p}(Q^n/ \cC_{k-1}) \xrightarrow{\pi_{hom}} H_{p}(Q^n/ \cC_{k})
\xrightarrow{\partial_{hom}} H_{p-1}(Q^n/ \cC_{k}) \rightarrow 0.
\end{equation}
\end{lemma}

\begin{proof}
By definition, $\partial_{chain}$ and $\delta_{cochain}$ are adjoint.

We also know from standard homology theory that the bilinear form $\langle \_ \, , \_ \, \rangle$ is well defined at the level of homology and cohomology groups. Using \cref{conecting_morphism_1,conecting_morphism_2}, we see that the connecting morphisms $\partial_{hom}$ and $\delta_{cohom}$ are adjoint at the level of homology and cohomology groups.

Using Lemma \ref{adjoints}, it is straightforward to prove that $\pi_{hom}$ and $i_{cohom}$ are adjoint because they correspond to $\pi_{chain}$ and $i_{cochain}$ on representatives. Similarly $\pi_{cohom}$ and $i_{hom}$ are adjoint.

In subsection \S \ref{cohomology basis} we have proved that $\pi_{cohom}$ is zero. Thus its adjoint $i_{hom}$ is also zero and the long exact sequences in homology break into pieces of short exact sequences.
\end{proof}

To summarise, we have obtained two short exact sequences adjoint to each other, one in homology and one in cohomology:

\vspace{-10pt}
\begin{align*}
&0 \rightarrow H_{p}(Q^n/ \cC_{k-1}) \xrightarrow{\pi_{hom}} H_{p}(Q^n/ \cC_{k})
\xrightarrow{\partial_{hom}} H_{p-1}(Q^n/ \cC_{k}) \rightarrow 0, \\
&0 \leftarrow H^{p}(Q^n/ \cC_{k-1}) \xleftarrow{i_{cohom}} H^{p}(Q^n/ \cC_{k})
\xleftarrow{\delta_{cohom}}  H^{p-1}(Q^n/ \cC_{k}) \leftarrow 0.
\end{align*}

\subsection{Product cycles in a generalized hemicube} \label{product cycles}

For cohomology groups, we were able to define cohomology bases with a preferred representative for each basis element. These preferred representative are canonical cocycles. We would like to do the same with homology groups. The preferred representatives for the elements of a homology basis are the soon to be defined product cycles.
 
Before we define product cycles, we need to define more formally translations at the level of
coordinates, faces and chains in $Q^n$ and in $Q^n/\cC_k$. We have already used translations at the level of coordinates and faces but we take the opportunity to be more formal now:

\begin{defn}[coordinate translation]
$$0 + 0 = 0 \qquad \qquad \qquad  \qquad \qquad 0 + 1 = 1$$
$$1 + 0 = 1 \qquad \qquad \qquad  \qquad \qquad 1 + 1 = 0$$
$$* + 0 = * \qquad \qquad \qquad  \qquad \qquad * + 1 = *$$
\end{defn}

\begin{defn}[face translation]
Let $f = f(1) \ldots  f(n) \in \{0, 1, *\}^n$ be a face of $Q^n$ and $y = y(1) \ldots  y(n) \in \{0, 1\}^n$ be a binary vector. We define $f + y$, the translate of $f$ by $y$, by coordinate-wise translation:
$$ \forall i \in \{1, \ldots  , n\}, \, (f + y)(i) = f(i) + y(i).$$
\end{defn}

\begin{defn}[chain translation]
Let $c = \bigoplus f \in C_p(Q^n)$ be a $p$-chain of $Q^n$ and $y = y(1) \ldots  y(n) \in \{0, 1\}^n$ be a binary vector. We define $c + y$, the translate of $c$ by $y$, by translation of every face of the $p$-chain:
$$c +  y = \bigoplus_{f\in c} (f + y).$$
\end{defn}

\noindent
Since the translation by $y$ is compatible with (commutes with) taking the quotient by a classical code $\cC_k$, we use the same definitions in $Q^n/\cC_k$. \\

Recall that the $p$-direction $I \subset \{1, \ldots, n\}$ of a $p$-face $f$ is the subset of coordinates where $f$ has a star: $I = \{ i \in \{1, \ldots  , n\} \, | \, f(i) = *\}$. In $Q^n$, there are $2^{n-p}$ $p$-faces having a given $p$-direction $I$. We name one of
them {\em the standard $p$-face with $p$-direction $I$} and denote it by
$f_{I}$.

\begin{defn}[standard $p$-face with $p$-direction $I$]
Let $I$ be a $p$-direction.
For every $i \in \{1, \ldots  , n\}$ , we define $s_i$ as the cardinality of $I \cap \{1, \ldots, i\}$. We define $f_I(i)$, the $i^{th}$ coordinate of $f_I$ to be $s_i$ modulo 2. \\
In $Q^n$, $f_I$ is the {\em standard $p$-face with $p$-direction $I$}. \\
In $Q^n/\cC_k$, $f_{I,k}$, the {\em standard $p$-face with $p$-direction $I$} is $\floor{f_I}$, the image  of $f_I$ under the projection $\Pi_k := \pi_k \circ \ldots  \circ \pi_1$.
\end{defn}

\noindent
For example with $n = 8$ and $p = 2$, the standard $2$-face with $2$-direction $\{3,7\} = \_ \_ * \_ \_ \_ * \_ $ is $0 0 * 1 1 1 * 0$ in $Q^n$. It is $\floor{0 0 * 1 1 1 * 0}$ in $Q^n/\cC_k$. \\

\begin{defn}[product chain]
For $x_1, \ldots  , x_k \in \mathbb{F}_2^n$ and $p_1, \ldots  , p_k \in \mathbb{N}$ such that $p_1 + \ldots  + p_k = p$, we define a \emph{product chain} in $C_p(Q^n)$ 
\[
\chain\tbinom{x_i}{p_i}_{1 \leq i \leq k}
\]
as follows:

\noindent
A $k$-tuple $(I_1, \ldots, I_k)$ of subsets of $\{1, \ldots, n\}$ is \emph{adapted to} 
$\tbinom{x_i}{p_i}_{1 \leq i \leq k}$
if it satisfies the following conditions:
\begin{itemize}
\item $\forall i, j \in \{1, \ldots, k\}, \; I_i \cap I_j = \varnothing$.
\item $\forall i \in \{1, \ldots, k\}, \; I_i \subset \text{Supp}(x_i)$.
\item $\forall i \in \{1, \ldots, k\}, \; |I_i| =p_i$. 
\end{itemize}

\noindent
The chain $\chain\tbinom{x_i}{p_i}_{1 \leq i \leq k}$ in $C_p(Q^n)$ is the sum of the standard $p$-faces $f_{I_1 \cup \ldots  \cup I_k}$ over every $k$-tuple $(I_1, \ldots, I_k)$ satisfying the above conditions.

\noindent
Similarly the chain $\chain_{k'}\tbinom{x_i}{p_i}_{1 \leq i \leq k}$ in $C_p(Q^n/C_{k'})$ is the sum of the standard $p$-faces $f_{I_1 \cup \ldots  \cup I_k, k'}$ over every $k$-tuple $(I_1, \ldots, I_k)$ satisfying the above conditions.
\end{defn}

Note that the sum is over $k$-tuples $(I_1, \ldots, I_k)$ and not over $p$-directions $I_1
\cup \ldots  \cup I_k$. It means that if a $p$-direction $I$ admits an even
number of adapted partitions $(I_1, \ldots, I_k)$, it actually doesn't belong to
$\chain\tbinom{x_i}{p_i}_{1 \leq i \leq k}$.

\medskip

\begin{lemma}
\label{product chain boundary}
The boundary of a product chain is:
\begin{align*}
&\partial\,\chain_{k'}\tbinom{x_i}{p_i}_{1 \leq i \leq k} \\
&= \bigoplus_{j=1}^k \; \chain_{k'}\tbinom{x_i}{p_i- \delta_{i,j}}_{1 \leq i \leq k}
\; \oplus \; \left(\chain_{k'}\tbinom{x_i}{p_i- \delta_{i,j}}_{1 \leq i \leq k}
 + x_j\right)
\end{align*}
\end{lemma}

\begin{proof}
Taking the boundary of a $p$-chain amounts to replacing each star of each of its $p$-faces by either a 0 or a 1.
Let $I_1 \subset \text{Supp}(x_1)$, \ldots  , $I_k \subset \text{Supp}(x_k)$
satisfy $|I_1| = p_1$, \ldots, $|I_k| = p_k$ and $I_i \cap I_j = \varnothing$
for every $i, j \in \{1, \ldots, k\}$. Let $f_{(I_1, \ldots, I_k)}$ be the
corresponding $p$-face of $\chain_{k'}\tbinom{x_i}{p_i}_{1 \leq i \leq k}$.
Choosing a star from $f_{(I_1, \ldots, I_k)}$ amounts to choosing $j \in \{1, \ldots, k\}$ and a star in $I_j$. It therefore yields $k$ intervals $\tilde{I_i}$ defined by $\tilde{I_i} = I_i$ if $i \neq j$ and $\tilde{I_j} = I_j \backslash \{i_*\}$ where $i_*$ is the coordinate of the chosen star. \\
Replacing $\{i_*\}$ by a zero or a one gives two translates of $f_{(I_1, \ldots, I_j \backslash \{i_*\}, \ldots,  I_k)}$. To each $\tilde{I_j}$ such that $(I_1, \ldots, \tilde{I_j}, \ldots, I_k)$ is adapted to 
$\tbinom{x_i}{p_i}_{1 \leq i \leq k}$
and such that $\tilde{I_j} = (I_j \backslash \{i_*\}) \cup \{i_{x_j}\}$ for $i_{x_j} \in \text{Supp}(x_j)$ correspond two other translates of $f_{(I_1, \ldots, I_j \backslash \{i_*\}, \ldots,  I_k)}$.
When summed, some of these translates cancel pairwise and we are left with
$f_{(I_1, \ldots, I_j \backslash \{i_*\}, \ldots,  I_k)} \oplus (f_{(I_1,
\ldots, I_j \backslash \{i_*\}, \ldots,  I_k)}+ x_j)$. \\
Summing over every possible $(I_1, \ldots, I_j \backslash \{i_*\}, \ldots,  I_k)$ finishes the proof.
\end{proof}

\begin{corol-defn}[product cycles]
For $c_1, \ldots, c_k \in \cC_k$, 
$\chain_k\tbinom{c_i}{p_i}_{1 \leq i \leq k}$ is a $p$-cycle of $C_p(Q^n/\cC_k)$.
We denote it by $\cyc_{p,k}\tbinom{c_i}{p_i}_{1 \leq i \leq k}$ and call these
cycles {\em product cycles}.
\end{corol-defn}

\begin{proof}
$$\forall j, \; \Pi_k \left( \chain\tbinom{c_i}{p_i - \delta_{i,j}}_{1 \leq i \leq k}
\right) = \Pi_k \left(\chain\tbinom{c_i}{p_i - \delta_{i,j}}_{1 \leq i \leq k}
+c_j\right) $$
\end{proof}

\begin{corol}
Let $C_{k-1}$ be a classical code with basis $(c_1, \ldots, c_{k-1})$ and $C_{k}$ be a classical code containing $C_{k-1}$ and with basis $(c_1, \ldots, c_{k})$. Let $\partial_{hom, p, k}$ be the homology group morphism corresponding to the classical codes $C_{k-1}$ and $C_{k}$. We have:
$$\partial_{hom,p,k} \left([\cyc_{p,k}\begin{pmatrix} c_i \\ p_i \end{pmatrix}_{1
\leq i \leq k}]\right) = [\cyc_{p-1,k}\begin{pmatrix} c_i \\ p_i - \delta_{i,k}
\end{pmatrix}_{1 \leq i \leq k}].$$
\end{corol}

\begin{proof}
\begin{align*}
&\partial ( \pi_k^{-1} ( \cyc_{p,k}\begin{pmatrix} c_i \\ p_i \end{pmatrix}_{1 \leq i \leq k} ) ) 
= \partial ( \pi_k^{-1} ( \Pi_k (\chain\begin{pmatrix} c_i \\ p_i \end{pmatrix}_{1 \leq i \leq k} ) ) ) \\
&=\partial ( \Pi_{k-1} (\chain\begin{pmatrix} c_i \\ p_i \end{pmatrix}_{1 \leq i \leq k} ) ) \\
&= \Pi_{k-1} (\chain\begin{pmatrix} c_i \\ p_i - \delta_{i,j} \end{pmatrix}_{1
\leq i \leq k}) \oplus \Pi_{k-1} (\chain\begin{pmatrix} c_i \\ p_i -
\delta_{i,j} \end{pmatrix}_{1 \leq i \leq k} + c_k) \\
&= i_k ( \Pi_{k-1} (\chain\begin{pmatrix} c_i \\ p_i - \delta_{i,j} \end{pmatrix}_{1 \leq i \leq k}) ) \\
&= i_k ( \cyc_{p-1,k}\begin{pmatrix} c_i \\ p_i - \delta_{i,k} \end{pmatrix}_{1
\leq i \leq k} ). 
\end{align*}

\noindent
We used Lemma \ref{product chain boundary} to derive line 3. \\
Recalling that $\partial_{hom,p,k}$ corresponds to $i_{chain,p,k}^{-1} \circ \partial_{chain,p,k} \circ \pi_{chain,p,k}^{-1}$ finishes the proof.
\end{proof}

\subsection{Homology basis for a generalized hemicube}

We are now ready to prove Theorem \ref{homology basis} by induction on $(p+k)$:

\begin{theorem} \label{homology basis}
Let $C_{k}$ be a classical code with basis $(c_1, \ldots, c_{k})$. $H_{p}(Q^n/ \cC_{k})$ has a
basis indexed by $k$-tuples $(p_1, \ldots, p_k)$ satisfying $p_1 + \ldots  + p_k
= p$ and such that each basis element is the homology class represented by the
product cycle $\cyc_{p,k}\tbinom{c_i}{p_i}_{1 \leq i \leq k}$.
\end{theorem}

\begin{proof}
The base case is straightforward. \\
We use the following short exact sequence and apply the induction
hypothesis to $H_{p}(Q^n/ \cC_{k-1})$ and $H_{p-1}(Q^n/ \cC_{k})$:

\begin{equation}
\nonumber
0 \rightarrow H_{p}(Q^n/ \cC_{k-1}) \xrightarrow{\pi_{hom}} H_{p}(Q^n/ \cC_{k})
\xrightarrow{\partial_{hom}} H_{p-1}(Q^n/ \cC_{k}) \rightarrow 0.
\end{equation}

Since $\pi_{chain,k} (\cyc_{p,k-1}\tbinom{c_i}{p_i}_{1 \leq i \leq k-1}) =
\cyc_{p,k}\tbinom{c_i}{p_i}_{1 \leq i \leq k}$ with $p_k = 0$, the basis of
homology classes of $H_{p}(Q^n/ \cC_{k-1})$ represented by product cycles is sent by $\pi_{hom}$ to a free family of homology classes represented by the product cycles of 
$H_{p}(Q^n/ \cC_{k})$ satisfying $p_k = 0$. \\

Since 
\[
\partial_{hom} ([\cyc_{p,k}\tbinom{c_i}{p_i}_{1 \leq i \leq k}]) =
[\cyc_{p-1,k}\tbinom{c_i}{p_i - \delta_{i,k}}_{1 \leq i \leq k}],
\]
the basis of homology classes of $H_{p-1}(Q^n/ \cC_{k})$ represented by product cycles has a free family of preimages by $\partial_{hom}$ represented by the product cycles of 
$H_{p}(Q^n/ \cC_{k})$ satisfying $p_k \neq 0$. 

The exactness of the short sequence implies that the concatenation of these two
free families forms a basis of $H_{p}(Q^n/ \cC_{k})$.
\end{proof}

\subsection{Cocycle minimum distance in a generalized hemicubic code} \label{cocycle md}

\begin{lemma}
For any product cycle $\cyc_{p,k}\tbinom{c_i}{p_i}_{1 \leq i \leq k}$, 
for any $y \in \mathbb{F}_2^n$, the translate  $\cyc_{p,k}\tbinom{c_i}{p_i}_{1
\leq i \leq k} + y$ is a cycle which belongs to the same homology class as 
$\cyc_{p,k}\tbinom{c_i}{p_i}_{1 \leq i \leq k}$.
\end{lemma}

\begin{proof}
The translate is a cycle since $\partial_{chain}$ and translation by $y$ commute. 

\noindent
To prove that translation doesn't alter the homology class we show that 
\[
\cyc_{p,k}\begin{pmatrix} c_i \\ p_i \end{pmatrix}_{1 \leq i \leq k} \oplus
(\cyc_{p,k}\begin{pmatrix} c_i \\ p_i \end{pmatrix}_{1 \leq i \leq k} + y)
\]
is a boundary. 
Equivalently we show that it is orthogonal to every cohomology class in
$H^{p}(Q^n/ \cC_{k})$. 

It is sufficient to consider the canonical cocycles representing a basis of
$H^{p}(Q^n/ \cC_{k})$. Observing that 
$\cyc_{p,k}\tbinom{c_i}{p_i}_{1 \leq i \leq k} \oplus
(\cyc_{p,k}\tbinom{c_i}{p_i}_{1 \leq i \leq k} + y)$ has exactly 0 or 2 $p$-faces per $p$-direction finishes the proof.
\end{proof}

Therefore each homology class of the product cycles basis of $H_{p}(Q^n/ C_{k})$
is represented by $2^{n-k}$ different cycles corresponding to the $2^{n-k}$
different translations $y \in \mathbb{F}_2^n / \cC_k$. Each $p$-face belongs to exactly $0$ or $2^p$ of the $2^{n-k}$ different cycles. This observation leads to the following proposition:

\begin{prop}
The cocycle minimum distance $D_{p,k}^{(cohom)}$, \text{i.e.} the minimum weight
of a cohomologically nontrivial $p$-cocycle in $\mathcal{C}^p(Q^n/\cC_k)$ satisfies:
$$ D_{p,k}^{(cohom)} = 2^{n-p-k}.$$
\end{prop}

\begin{proof}
Let $\eta^{p,k}$ be a cohomologically nontrivial $(p,k)$-cocycle. $\eta^{p,k}$ is not orthogonal to at least one product cycle representing an element of the basis of $H_{p}(Q^n/ C_{k})$. Therefore $\eta^{p,k}$ is not orthogonal to any of the $2^{n-k}$ different cycles obtained by translating this product cycle. Since each $p$-face of $\eta^{p,k}$ belongs to at most $2^p$ translated product cycles, $\eta^{p,k}$ contains at least $2^{n-k-p}$ $p$-faces. \\
Moroever the value $2^{n-p-k}$ is attained by canonical cocycles.
\end{proof}

\subsection{Cycle minimum distance in a generalized hemicubic code} \label{cycle md}

\begin{prop}
The cycle minimum distance $D_{p,k}^{(hom)}$, \text{i.e.} the minimum weight of
a homologically nontrivial $p$-cycle in $C_p(Q^n/\cC_k)$ satisfies:
$$ D_{p,k}^{(hom)} = {d \choose p}.$$
\end{prop}

\begin{proof}
We prove by induction on $(p+k)$ that a homologically nontrivial $(p,k)$-cycle is not orthogonal to at least ${d \choose p}$ canonical $(p,k)$-cocycles. Since canonical cocycles are disjoint, the value of the cycle minimum distance follows immediately.

The base case is straightforward. 

Let $\eta_{p,k}$ be a cycle representing a nontrivial homology class $h_{p,k} \in H_{p}(Q^n/ C_{k})$: $h_{p,k} = \floor{\eta_{p,k}}$.

\medskip

\noindent
\underline{first case:} $\partial_{hom}(h_{p,k}) = 0$ in $H_{p-1}(Q^n/ \cC_{k})$
for at least one decomposition $\cC_k = \cC_{k-1} \cup (\cC_{k-1} + c_k)$. \\

Then there exists a nontrivial homology class $h_{p,k-1} \in H_{p}(Q^n/
C_{k-1})$ such that $h_{p,k} = \pi_{hom}(h_{p,k-1})$. Let $\eta_{p,k-1}$ be a $(p,k-1)$-cycle representing $h_{p,k-1}$. \\
By the induction hypothesis there are ${d \choose p}$ canonical
$(p,k-1)$-cocycles not orthogonal to $\eta_{p,k-1}$. Let $\cocyc^{p,k-1}$ be such a canonical cocycle. Since $\pi_{chain}$ and $i_{cochain}$ are adjoint:

\begin{align*}
\langle i_{cochain}^{-1} (\cocyc^{p,k-1}) \, , \, \eta_{p,k} \rangle 
= \quad &\langle \cocyc^{p,k-1} \, , \, \pi_{chain}^{-1} (\eta_{p,k}) \rangle \\
= \quad &\langle \cocyc^{p,k-1} \, , \, \eta_{p,k-1} \rangle \\
= \quad &1.
\end{align*}

Therefore applying $i_{cochain}^{-1}$ to the ${d \choose p}$ canonical
$(p,k-1)$-cocycles not orthogonal to $\eta_{p,k-1}$ yields ${d \choose p}$
canonical $(p,k)$-cocycles not orthogonal to $\eta_{p,k}$. The induction step is proved in this case. 

\bigskip

\noindent\underline{second case:} $\partial_{hom}(h_{p,k}) \neq 0$ in $H_{p-1}(Q^n/
\cC_{k})$ for every decomposition $\cC_k = \cC_{k-1} \cup (\cC_{k-1} + c_k)$. \\

By definition of $\partial_{hom}$, any preimage $i_{chain}^{-1} \circ
\partial_{chain} \circ \pi_{chain}^{-1} (\eta_{p,k})$ represents $\partial_{hom} (h_{p,k})$. 

By the induction hypothesis there exists ${d \choose p-1}$ distinct canonical
$(p-1,k)$-cocycles orthogonal to $i_{chain}^{-1} \circ \partial_{chain} \circ
\pi_{chain}^{-1} (\eta_{p,k})$. Let $\cocyc^{p-1,k}$ be such a cocycle. Any
preimage $i_{cochain}^{-1} \circ \delta_{cochain} \circ \pi_{cochain}^{-1}
(\cocyc^{p-1,k})$ is a $(p,k)$-cocycle orthogonal to $\eta_{p,k}$: 

\begin{align*}
\langle i_{cochain}^{-1} \circ \delta_{cochain} \circ \pi_{cochain}^{-1}
(\cocyc^{p-1,k}) \, , \,  \eta_{p,k} \rangle 
= \quad &\langle \delta_{cochain} \circ \pi_{cochain}^{-1} (\cocyc^{p-1,k}) \, ,
\, \pi_{chain}^{-1} (\eta_{p,k}) \rangle \\
= \quad &\langle \pi_{cochain}^{-1} (\cocyc^{p-1,k}) \, , \, \partial_{chain}
\circ \pi_{chain}^{-1} (\eta_{p,k}) \rangle \\
= \quad &\langle \cocyc^{p-1,k} \, , \, i_{chain}^{-1} \circ \partial_{chain}
\circ \pi_{chain}^{-1} (\eta_{p,k}) \rangle \\
= \quad &1. 
\end{align*}

Let us count the number of canonical $(p,k)$-cocycles $i_{cochain}^{-1} \circ \delta_{cochain} \circ \pi_{cochain}^{-1} (\cocyc^{p-1,k})$ that we can construct from the ${d \choose p-1}$ distinct canonical $(p-1,k)$-cocycles $\cocyc^{p-1,k}$. \\
Since $i_{cochain}$ is a bijection, $i_{cochain}^{-1}$ is uniquely defined. But $\pi_{cochain}^{-1} (\cocyc^{p-1,k})$ can be any preimage of $\cocyc^{p-1,k}$ by $\pi_{cochain}$. We use the same technique as in the construction of the cohomology basis represented by canonical cocycles. 

The $k^{th}$ element of the basis of the classical code $c_k$ has weight at
least $d$. Let $I$ be the $(p-1)$-direction of the canonical cocycle $\cocyc^{p-1,k}$. At least $(d-p+1)$ coordinates are in $\text{Supp}(c_k) \backslash I$. Denoting by $j$ one of these $(d-p+1)$ coordinates, the $(p-1)$-cochain obtained by only keeping the $(p-1)$-faces of $\cocyc^{p-1,k}$ having a 0 at coordinate $j$ is a preimage of $\cocyc^{p-1,k}$ by the $\pi_{cochain}$ associated to a decomposition $\cC_k = C_{k-1} \cup (C_{k-1} \oplus c_k)$ such that $\forall x \in C_{k-1}, x_j = 0$. Applying $\delta_{cochain}$ to this cochain amounts to replacing this 0 at coordinate $j$ of every $(p-1)$-face by a $*$ and yields $\cocyc^{I \cup \{j\},p,k-1}$. Applying $i_{cochain}^{-1}$ gives $\cocyc^{I \cup \{j\},p,k}$. 

With this procedure each canonical $(p,k)$-cocycle $\cocyc^{I \cup \{j\},p,k}$
has been counted at most $p$ times. We have therefore constructed at least
$\frac{d-p+1}{p} {d \choose p-1} = {d \choose p}$ distinct canonical
$(p,k)$-cocycle orthogonal to $\eta_{p,k}$. The induction step is proved in this case too. \\

\noindent
Moreover the value $d \choose p$ is attained by the product cycles
$\cyc_{p,k}\tbinom{c_i}{p_i}_{1 \leq i \leq k}$ such that $p_1 = p$, $p_{i \neq 1} = 0$ and $c_1$ has weight $d$.
\end{proof}

We have thus established:
\param*

\section{Local testability}
\label{sec:LTC}

The goal of this section is to study the local testability of hemicubic codes. We first establish in \ref{subsec:LTC} that the one-qubit hemicubic code is locally testable, before discussing generalized hemicubic codes in \ref{subsec:LTC2}.

\subsection{Case of the 1-qubit hemicubic code}
\label{subsec:LTC}

We first prove the local testability of the hemicube code.
\LTC*

This improves upon Hastings' construction \cite{has16} obtained by taking the product of two $n$-spheres and which displays soundness $s = \Theta\left(\frac{1}{\log^2 N}\right)$. (In Ref.~\cite{has16}, the notion of soundness does not include a normalization by the logarithmic weight of the generators.)
We leave it as an important open question whether the bounds of Theorem \ref{thm:ltc} are tight or not. As far as we know, it may be possible to improve the $\log N$ to $\Theta(1)$. As we will mention later, this would imply that the generalized hemicubic code obtained as the quotient of the cube by a code of dimension 2 would also display local testability.

In this section, we will work in a symmetrized version of the hemicubic code:
instead of describing a $p$-face of the hemicube by an equivalence class of the form $\{x, \bar{x}\}$, we consider the chain $x + \bar{x}$ over the Hamming cube. In the language of the previous section, we work with $i_p(E)$ rather than directly with a $p$-chain $E$. As long as all the considered sets $S$ are symmetric, i.e., are in the image of $i_p$, there should not be any risk of confusion. In particular, any symmetric set $S$ of $p$-faces in the Hamming cube corresponds to a set of $|S|/2$ qubits.

The local testability of the hemicubic code is a consequence of Lemmas \ref{thm:chain} and
\ref{thm:cochain} that we state now: we use the notation $\|\;\|$ for the
Hamming weight, or number of cells in a chain.

We say that a $p$-chain $X$ is a \emph{filling} of $Y$ if $\partial X=Y$. We say that a $p$-cochain $X$ is a \emph{cofilling} of $Y$ if $\delta X = Y$.

\begin{lemma} \label{thm:chain}
Let $E$ be a $p$-chain of $C_p^n=C_p(Q^n/\cC_r)$, where $\cC_r ={ 00\ldots0, 11\ldots 1}$ is the repetition code. Then there exists a $p$-chain $F$ which is a filling of $\partial E$, satisfying $\partial F = \partial E$, such that 
\begin{align*}
\|F\| \leq c_{n,p} \|\partial E\|,
\end{align*}
with 
\begin{align*}
c_{n,p}= \frac{(n-p+1)(n-p)}{2p} \sum_{m=n-p+1}^{n} \frac{1}{m}.
\end{align*}
\end{lemma}

\begin{lemma} \label{thm:cochain}
Let $E$ be a $p$-cochain of $C_p^n$. Then there exists a $p$-cochain $F$ which
is a cofilling of $\delta E$, satisfying $\delta F = \delta E$, such that 
\begin{align*}
\|F\| \leq c_{n,p}'  \|\delta E\|,
\end{align*}
with 
\begin{align*}
c'_{n,p}=(n-p-1) \sum_{m=n-p}^n \frac{1}{m}.
\end{align*}
\end{lemma}

In particular, the following upper bounds hold for $c_{n,p}$ and $c'_{n,p}$ (obtained by bounding each term in the sum by the largest term):
\begin{align*}
c_{n,p} \leq \frac{n-p}{2}, \quad c'_{n,p} \leq p+1. 
\end{align*}

It is straightforward to translate these results in the language of quantum codes. Indeed, given an arbitrary Pauli error $E = (E_X, E_Z)$ where $E_X$ and $E_Z$ represent the supports of the $X$-type and $Z$-type errors, the syndrome of $E$ is given by the pair $(\partial E_X, \delta E_Z)$, where $E_X$ and $E_Z$ are interpreted as a $p$-chain, and a $p$-cochain respectively.
To compute the soundness of the quantum code, one needs to lower bound the ratio:
\begin{align*}
\min_{(E_X,E_Z)} \frac{\|\partial E_X\| + \|\delta E_Z\|}{\|[E_X]\| + \|[E_Z]\|}
\geq \min \left\{\min_{E_X} \frac{\|\partial E_X\| }{\|[E_X]\|}, \min_{E_Z}
\frac{ \|\delta E_Z\|}{ \|[E_Z]\|}  \right\},
\end{align*}
where the minimum is computed over all errors with a nonzero syndrome, i.e., for $p$-chains $E_X$ which are not a $p$-cycle and $p$-cochains $E_Z$ which are not a $p$-cocycle. In these expressions, we denote by $[E]$ the representative of the equivalence class of error $E$, with the smallest weight. Indeed, recall that two errors differing by an element of the stabilizer group (that is, by a boundary or a coboundary) are equivalent. 
The fact that one considers $[E]$ instead of $E$ makes the analysis significantly subtler in the quantum case than in the classical case. A solution is to work backward (as was also done in \cite{has16}): start with a syndrome and find a small weight error giving rise to this syndrome. 
This is essentially how Lemmas \ref{thm:chain} and \ref{thm:cochain} proceed to bound each term:
\begin{align}\label{lem2ltc}
\min_{E_X, \partial E_X \neq 0} \frac{\|\partial E_X\| }{\|[E_X]\|} \geq
\frac{1}{c_{n,p}}, \quad \min_{E_Z, \delta E_Z \neq 0} \frac{ \|\delta E_Z\|}{
\|[E_Z]\|}  \geq \quad \frac{1}{c'_{n,p}}. 
\end{align}
Indeed, one should think of the $p$-chain $\partial E$ in Lemma \ref{thm:chain}
as the syndrome associated to error $E$, and the lemma shows the existence of an error $F$ with small weight (possibly different from $E$) with the same syndrome. Lemma \ref{thm:cochain} provides the equivalent result for the other type of errors. 
The soundness in Theorem \ref{thm:ltc} then results from $c_{n,p}, c'_{n,p} = O(\log N)$, where the second logarithmic factor (yielding a final soundness of $1/\log^2 N$) comes from the additional normalization by the generator weights. 

Before establishing these two lemmas, we recall two similar results due to
Dotterrer and holding in the hypercube instead of the hemicube, that is, without taking the quotient by the repetition code \cite{dot16}.

\begin{lemma}[\cite{dot16}] \label{lem:dott1}
Let $z$ be a $(p-1)$-dimensional $\mathbbm{F}_2$-cycle in the $n$-dimensional
cube $Q^n$. There exists a $p$-chain $y$ such that $\partial y = z$ and 
\begin{align*}
\|y \| \leq \frac{n-p+1}{2p} \|z\|.
\end{align*}
\end{lemma}

\begin{lemma}[Proposition 8.2.1 of \cite{dot13}] \label{lem:dott2}
Let $z$ be a $(p+1)$-dimensional $\mathbbm{F}_2$-cocycle in the $n$-dimensional
cube $Q^n$. There exists a $p$-cochain $y$ such that $\delta y = z$ and 
\begin{align*}
\|y \| \leq  \|z\|.
\end{align*}
\end{lemma}

The constants in Lemmas \ref{lem:dott1} and \ref{lem:dott2} are tight \cite{dot13}.
We don't know, however, if it is also the case of the constants in Lemmas \ref{thm:chain} and \ref{thm:cochain}: in fact it is not even clear that the constants have to be worse than those of \ref{lem:dott1} and \ref{lem:dott2} since the examples saturating these bounds are not allowed in the symmetric (quantum) case. For instance, the cycles of the cube $Q^n$ that saturate the bound of Lemma \ref{thm:chain} are symmetric, meaning that they disappear in the case of the hemicubic code.

\begin{proof}[Proof of Lemma \ref{thm:chain}]

We will prove the claim by recurrence over both $p$ and $n$. 

Similarly to Dotterrer in \cite{dot16}, we divide the cube in three parts: we first choose a coordinate that we call the ``cut'' and partition the faces depending on their value, 0, 1 or $*$, for the cut. Later, we will perform an optimization over the choice of cut, but in the following, we consider a cut along the first coordinate to fix the notations.

Let us define the chain $Z = \partial E$ and decompose it as 
\begin{align*}
Z &= 0Z_0 \, \oplus \, * Z_* \, \oplus \, 1Z_1
\end{align*}
where $Z_0$, $Z_1$ are chains of $C_{p-1}^{n-1}$ and $Z_*$ is a chain of $C_{p-2}^{n-1}$.
Since $Z$ is a cycle, we have that $\partial Z=0$ which implies
\begin{align}\label{eqn:cond1}
Z_* = \partial Z_0 = \partial Z_1.
\end{align}
We can define the chains $E_0, E_1$ and $E_*$ in an analogous fashion, \emph{via} $E=0E_0 \, \oplus \, * E_* \, \oplus \, 1E_1$, and from $\partial E=Z$, we infer in particular that $Z_* = \partial E_*$. Applying the induction hypothesis to $Z_*$ gives a $(p-1)$-chain $F_*$ such $\partial F_* = \partial E_*$ and
\begin{align*}
\|F_*\| \leq c_{n-1, p-1} \|\partial E_*\|.
\end{align*}
Observe now that $Z_0 \, \oplus \, F_*$ is a cycle: indeed 
\begin{align*}
\partial (Z_0 \, \oplus \, F_*) = \partial Z_0 \, \oplus \, Z_* = 0,
\end{align*}
from Eq.~\eqref{eqn:cond1}.
Applying Lemma \ref{lem:dott1} for the standard hypercube, we can find a $p$-chain $F_0$ (that may not be symmetric) such that 
\begin{align*}
\partial F_0 = Z_0 \, \oplus \, F_*
\end{align*}
and 
\begin{align} \label{eqn:36}
\|F_0\| \leq \frac{n-p}{2p} \|Z_0 \, \oplus \, F_*\|.
\end{align}
Define $F_1 = \overline{F_0}$ so that $\|F_1\| = \|F_0\|$ and $\partial F_1 = Z_1 \, \oplus \, F_*$.
We claim that the symmetrized chain $F = 0F_0 \, \oplus \, *F_* \, \oplus \, 1F_1$ satisfies the conditions of the theorem.
First, $F$ is a filling of $\partial E$:
\begin{align*}
\partial F &= 0( \partial F_0 \, \oplus \, F_*) + 1 (\partial F_1 \, \oplus \, F_*) \, \oplus \, * \partial F_*\\
&= 0 Z_0 \, \oplus \, 1 \overline{Z_0} \, \oplus \, * Z_* = \partial E.
\end{align*}
Second
\begin{align*}
\|F\| &= 2 \|F_0\| + \|F_*\|\\
&\leq 2 \frac{n-p}{2p} \|Z_0 \, \oplus \, F_*\| +  \|F_*\| & \text{from Eq.~\eqref{eqn:36}}\\
& \leq \frac{n-p}{p} \|Z_0\| + \frac{n}{p} \|F_*\| & \text{from triangle inequality}\\
& \leq \frac{n-p}{p} \|Z_0\| + \frac{n}{p} c_{n-1,p-1} \|Z_*\| & 
\end{align*}
Let us minimize the size of $F$ over the choice of the cut. In particular, the minimal value of $\|F\|$ is not larger than the expectation over the cut choice, when the coordinate of the cut is chosen uniformly at random. This expectation is easily computed by noticing that $\mathbbm{E} \|Z_0\| = \frac{n-p+1}{2n}\|Z\|$.
To see this, observe that there are $2n$ possible choices of cut: $n$ choices of coordinates and 2 choices to define the $0$ and $1$ orientation. Then each $(p-1)$-face is overcounted $n-(p-1)$ times because it lies in $n-(p-1)$ many faces of dimension $n-1$.
In addition, we get $\mathbbm{E} \|Z_*\| = \|Z\| -\mathbbm{E} \|Z_0\|-\mathbbm{E} \|Z_1\|= \frac{p-1}{n} \|Z\|$.
Let us denote by $F$ the chain of minimum size when optimizing over the cut choice. We have:
\begin{align*}
\frac{\|F\|}{\|Z\|} & \leq  \frac{(n-p-1)(n-p)}{2np} + \frac{n}{p} c_{n-1,p-1} \frac{p-1}{n}.
\end{align*}
In particular, this establishes that we can take
\begin{align*}
c_{n,p}  =  \frac{(n-p-1)(n-p)}{2np} + \frac{p-1}{p} c_{n-1,p-1}.
\end{align*}
The base case, $c_{n,1} = \frac{n-1}{2}$, differs from the value $\frac n 2$ that one would obtain in the cube with the assumption that the cycle is the boundary of a symmetric $(p+1)$-chain.
The recurrence relation can be solved as follows:
\begin{align*}
c_{n,p} &=\frac{(n-p-1)(n-p)}{2np} + \frac{p-1}{p} c_{n-1,p-1}\\
&= \frac{(n-p-1)(n-p)}{2} \left( \frac{1}{np} + \frac{1}{(n-1)p}  + \frac{p-2}{p-1} c_{n-2,p-2} \right)\\
&= \frac{(n-p-1)(n-p)}{2p} \left( \frac{1}{n} + \frac{1}{n-1} + \ldots \frac{1}{n-p+1} \right).
\end{align*}
This establishes the result.
\end{proof}

\begin{proof}[Proof of Lemma \ref{thm:cochain}]
We proceed in a similar way and establish the claim by recurrence over $n$ and $p$.
We pick an arbitrary coordinate (a cut in the language of Dotterrer) and denote $Z = \delta E = 0 Z_0 \, \oplus \, 1 Z_1 \, \oplus \, * Z_*$. The cofilling $F$ of $Z$ is defined as $F = \tilde{D}_{n,p}(Z)$ recursively as follows: 
\begin{align*}
F &= \tilde{D}_{n,p} (Z)\\
&=: 0 D_{n-1,p}(Z_0) \, \oplus \, 1 \overline{D_{n-1,p}(Z_0)}\, \oplus \, * \tilde{D}_{n-1 p-1} (D_{n-1,p}(Z_0) \,\oplus \, \overline{D_{n-1,p}(Z_0)} \, \oplus \, Z_*).
\end{align*}
where $D_{n-1,p}(Z_0)$ is the cofilling of the cocycle $Z_0$ obtained by Dotterrer's algorithm (i.e, the cofilling promised by Lemma \ref{lem:dott2}). Here, $\tilde{D}_{n-1,p-1}$ is the symmetric cofillings with parameters $n-1$ and $p-1$ given by the induction hypothesis.
Exploiting the result of Lemma \ref{lem:dott2}, we obtain that $F$ has size:
\begin{align}\label{eq64}
\|F\| &\leq 2 \|Z_0\| + c'_{p-1,n-1} (2 \|Z_0\| \, + \, \|Z_*\|).
\end{align}
Averaging over the choice of cut,
\begin{align*}
\mathbbm{E} \|F\|  &\leq 2 \frac{\mathbbm{E} \|Z_0\|}{2n} + c'_{p-1,n-1} \|Z\|\\
& \leq \left(\frac{n-p}{n} + c'_{p-1,n-1} \right) \|Z\|
\end{align*} 
which yields 
\begin{align*}
c'_{p,n} &\leq \frac{n-p}{n} + c'_{p-1,n-1}.
\end{align*}
The recurrence is easily solved:
\begin{align*}
c'_{p,n} &\leq (n-p) \left( \frac{1}{n-p+1} + \cdots \frac{1}{n}\right).
\end{align*}
The base case is $c'_{n,1} = \frac{1}{n}$.
This establishes the claim. 
\end{proof}

\subsection{Local testability of generalized hemicubic codes}
\label{subsec:LTC2}

In this subsection, we show that the same proof strategy as above can be applied to deal with quotients of the cube by linear codes of dimension $k=2$. Essentially the only change is that the recurrence now requires a bound on the soundness of the hemicubic code instead of a bound on the soundness of the standard cube. Because our bound on the former is worse by a factor $\log N$, we will not be able to control the soundness of the generalized hemicubic code as much as we would like. 

We now illustrate this point in the case of cycles and prove the following bound.

\begin{lemma} \label{thm:chain2}
Let $\cC = [n,2,d]$ be a linear code of dimension 2. Let $E$ be a $p$-chain of $C_p^n=C_p(Q^n/\cC)$. Then there exists a $p$-chain $F$ which is a filling of $\partial E$, satisfying $\partial F = \partial E$, such that 
\begin{align*}
\|F\| \leq c_{n,p}^{(2)} \|\partial E\|,
\end{align*}
with 
\begin{align*}
c_{n,p}^{(2)}= O(p!).
\end{align*}
\end{lemma}
We assume here that for any coordinate, there exists a codeword of $\cC$ with bit value 1 on this coordinate. If this is not the case, one can work in a Hamming cube of smaller dimension by forgetting this coordinate.

In the same way as before, we will choose to work in the standard Hamming cube, but restricting ourselves to sets (or chains, or cochains) of the form $\{x+\cC, y+\cC, \ldots\}$, i.e., sets $S$ such that $x \in S$ implies $x + c \in S$ for any codeword $c \in \cC$. 
In other words, we work with sets (or chains) of the form $i_p(E) = \bigoplus_{e \in E, c \in \cC} (e+c)$.

\begin{proof}[Proof of Lemma \ref{thm:chain2}]

Let us consider a $(p-1)$-chain $Z=\partial E$ corresponding to the boundary of an arbitrary $p$-chain $E$, symmetric with respect to the code $\cC$. Recall that this means that for any $c\in \cC$, it holds that $E+c= E$.
Without loss of generality, we can choose some $E_0$ and $E_*$, which are sets of $p$ and $(p-1)$-faces of the $(n-1)$-dimensional Hamming cube, such that
\begin{align*}
E = \bigoplus_{c \in \cC}\left( (0E_0 + c) \oplus (* E_* + c) \right).
\end{align*}
As before, this describes a partition with respect to the value of the symbol (either an element of $\F_2$ or a star) on the special coordinate called ``cut''.
That we can take $E_1$ to be empty is without loss of generality since we assumed that there are codewords of $C$ with bit value 1 for the cut. 
The boundary of $E$ is $Z = \partial E$ and our goal is to find a small symmetric filling $F$ such that $\partial F = \partial E$.
We will prove the result by induction on $n$ and $p$ by showing the existence of a map $\tilde{D}_{n,p}$ such that $\partial(\tilde{D}_{n,p} (\partial E)) = \partial E$ and $\|\tilde{D}_{n,p} (\partial E)\| \leq c_{n,p}^{(2)} \|\partial E\|$.

Let $1\alpha$ be a codeword of $\cC$ with bit value 1 on the cut ($\alpha$ is a word of length $n-1$). Again, as before, we pretend that the cut corresponds to the first coordinate to fix the notations. Let $A$ be the subcode of $\cC$ consisting of all codewords with bit value 0 on the cut. This yields a partition of $\cC$ as
\begin{align*}
\cC = A \cup (\alpha + A),
\end{align*}
where the set $A$ only contains codewords with bit value 0 on the cut, and $A+\alpha$ codewords with bit value 1. 
With this notation, we have:
\begin{align*}
E &= \bigoplus_{a \in A} 0(E_0 + a) \oplus 1(E_0 + \alpha +a )  \oplus*(E_* + a) \oplus  *(E_* + \alpha +a )\\
\partial E &=  \bigoplus_{a \in A} 0\big((\partial E_0 + a)\oplus (E_*+a) \oplus (E_*+\alpha +a ) \big) \oplus * \big((\partial E_* + a) \oplus(\partial E_* + \alpha +a ) \big)\\
& \quad \oplus 1 \big( (\partial E_0 + \alpha +a) \oplus (E_*+a) \oplus  (E_*+\alpha+a) \big)\\
&= 0 Z_0 \oplus 1 Z_1 \oplus * Z_*,
\end{align*}
with 
\begin{align*}
Z_0 &=\bigoplus_{a \in A}(\partial E_0 + a)\oplus (E_*+a) \oplus (E_*+\alpha +a ) \\
Z_1 &= \bigoplus_{a \in A}(\partial E_0 + \alpha +a) \oplus (E_*+a) \oplus  (E_*+\alpha+a) =Z_0+\alpha \\
Z_* &=\bigoplus_{a \in A}(\partial E_* + a) \oplus(\partial E_* + \alpha +a ).
\end{align*}
In particular, $Z_*$ is a boundary symmetric with respect to the shortened code obtained by forgetting the coordinate corresponding to the cut in $\cC$ and one can apply the induction hypothesis to obtain a small filling $F_* = \tilde{D}_{n-1,p-1}(Z_*)$ of size 
\begin{align*}
\|F_*\| \leq c_{n-1,p-1}^{(2)} \|Z_*\|.
\end{align*}

Let us observe that $Z_0 + F_*$ is a cycle. Indeed, 
\begin{align*}
\partial(Z_0 + F_*) &= \bigoplus_{a \in A}(\partial E_*+a) \oplus (\partial E_*+\alpha +a ) + \partial F_* = 0. 
\end{align*}
Since it is a cycle, it is a boundary in the quotient of $(n-1)$-dimensional Hamming cube code $A$. Applying the construction of Lemma \ref{thm:chain} to this boundary yields a filling $F_0$ symmetric with respect to $A$ (i.e. $F_0 + a = F_0$ for any $a\in A$) satisfying:
\begin{align*}
\partial F_0& = Z_0 \,\oplus \, \tilde{D}_{n-1,p-1}(Z_*)\\
\|F_0 \| & \leq c_{n-1,p}  \|Z_0 \,\oplus \, \tilde{D}_{n-1,p-1}(Z_*)\| \leq \frac{n-p+1}{2} \left(\| Z_0\| + c_{n-1,p-1}^{(2)} \|Z_*\|\right),
\end{align*}
where the factor $c_{n-1,p}  = \frac{n-p+1}{2}$ results from our bound on the size of a symmetric filling with respect to 
the repetition code. Note that some coordinates are likely stuck with the value 0 in the code $A$, and one might expect a better factor in that case, but we don't consider this possible improvement in the following.
Define $F_1 = F_0 +\alpha$. It is a filling of $Z_1$ since $Z_1 = Z_0+\alpha$. Moreover, the assumption on the minimum distance of $C$ implies that $\|F_0\|=\|F_1\|$.
We finally define $F = 0F_0 \oplus * F_* \oplus 1 F_1$ which satisfies $\partial F = \partial E$ by construction. 

As before, $\mathbbm{E} \|Z_0\| = \frac{n-p+1}{2n}\|Z\|$ and $\mathbbm{E} \|Z_*\| = \frac{p-1}{n} \|Z\|$ and therefore, we can take
\begin{align}\label{eq:rec2}
c_{n,p}^{(2)} & \leq 2c_{n-1,p} \left(\frac{\mathbbm{E} \| Z_0\|}{\|Z\|} + c_{n-1,p-1}^{(2)}\frac{ \mathbbm{E} \|Z_*\|}{\|Z\|} \right) +  c_{n-1,p-1}^{(2)} \frac{ \mathbbm{E} \|Z_*\|}{\|Z\|}. 
\end{align}
This is in fact the same recurrence relation as before (in Lemma \ref{thm:chain}), but with the value of $c_{n-1,p}^{(0)} = \frac{n-p}{2p}$ replaced by $c_{n-1,p}^{(1)} := c_{n-1,p}$.
We claim that 
\begin{align*}
c_{n,p}^{(2)} = \frac{(n-p+1)^2}{2} \sum_{\ell = 0}^p (n-p+2)^{\ell} \frac{(n-\ell+1)!p!}{n! (p-\ell)!}
\end{align*}
is a valid solution to this recurrence.

Upper bounding every term in the sum by the largest one corresponding to $\ell = p$, we get
\begin{align*}
c_{n,p}^{(2)}& \leq \frac{(n-p+1)^2}{2} (p+1)  (n-p+2)^{p} \frac{(n-p+1)!p!}{n!}\\
&\approx \frac{(n-p)^p}{\tbinom{n}{p}}= O(p!).
\end{align*}
\end{proof}

It is quite striking that the resulting bound on the soundness is much worse when taking the quotient by a classical code of dimension 2 rather than by the repetition code. In particular, the resulting soundness becomes only $1/\mathrm{poly}(N)$ instead of $1/\log N$. The source of this discrepancy is easily located in Eq.~\eqref{eq:rec2}, where we injected the value of the soundness for the hemicubic code instead of the soundness of the standard cube. 
In particular, if we could establish that the hemicube code had a similar soundness (or even better) than the standard cube, then the proof above would immediately imply that the generalized hemicube code has soundness $1/\mathrm{polylog}(N)$. This would provide the first example of quantum code of exponential length displaying local testability.

A similar analysis can be performed for cofillings but again, it only provides a bound for the soundness scaling inverse polynomially with $N$. 
\begin{lemma} \label{thm:cochain2}
Let $\cC = [n,2,d]$ be a linear code of dimension 2. Let $E$ be a $p$-cochain of $C_p^n=C_p(Q^n/\cC)$. Then there exists a $p$-cochain $F$ which is a cofilling of $\delta E$, satisfying $\delta F = \delta E$, such that 
\begin{align*}
\|F\| \leq c_{n,p}^{(2)'} \|\partial E\|,
\end{align*}
with 
\begin{align*}
c_{n,p}^{(2)'}= O(p!).
\end{align*}
\end{lemma}

\begin{proof}[Proof of Lemma \ref{thm:cochain} for arbitrary codes]

Using the same notations as in the previous subsection, we start with an arbitrary $p$-cochain $E$ which we write
\begin{align*}
E = \bigoplus_{c \in \cC}\left( (0E_0 + c) \oplus (* E_* + c) \right),
\end{align*}
with respect to an arbitrary cut.
Choosing a codeword $\alpha$ with bit value $1$ on the cut, and denoting by $A$ the subcode of $\cC$ consisting of all the words with bit value 0 on the cut, we obtain:
\begin{align*}
E &= \bigoplus_{a \in A} 0(E_0 + a) \oplus 1(E_0 + \alpha +a )  \oplus*(E_* + a) \oplus  *(E_* + \alpha +a )\\
\delta E &= 0 Z_0 + 1Z_1 + *Z_*.
\end{align*}
We again proceed by induction. Let us denote by $D_{n,p}$ the application
promised by Lemma \ref{thm:cochain} and by $\tilde{D}_{n,p}$ the application
promised by the present lemma (yielding a symmetric cofilling), and defined by induction. The latter application preserves the symmetry of the cochain with respect to $\cC$, while this is not necessarily the case of $D_{n,p}$, which only preserves the symmetry with respect to $A$.

We define the symmetric cofilling of $\delta E$ by
\begin{align*}
F &= \tilde{D}_{n,p}(\delta E)\\
& := * \tilde{D}_{n-1,p-1} \left(D_{n-1,p}\left ( \bigoplus_{a \in A}(E_0+a) \oplus (E_0 +\alpha +a)\right)+  \sum_{a \in A}(E_*+a) \oplus (E_* +\alpha +a)\right)  \\
&\quad  \oplus 0 D_{n-1,p}\left ( \bigoplus_{a \in A}(E_0+a)\right)  + 1 D_{n-1,p}\left( \bigoplus_{a \in A}(E_0+\alpha+a)\right).
\end{align*}
One can check that $\delta F = \delta E$ and that $F$ is symmetric with respect to code $\cC$.
Bounding the size of $F$ is similar to the proof in the case of the repetition code. Indeed, recalling that $\|D_{n-1,p}(X)\| \leq (p+1) \|X\|$ for any cocycle $X$ of the hemicube, we obtain
\begin{align*}
\|F\| & \leq c'_{n-1,p-1} \left( (p+1) \|Z_0\| +(p+1) \|Z_1\| + \|Z_*\| \right) +(p+1) \|Z_0\| + (p+1)\|Z_1\|,
\end{align*}
which is identical to Eq.~\eqref{eq64}, except for the extra factors $(p+1)$. 
As before, solving the recurrence yields $c_{n,p}^{(2)'}= O(p!)$.
\end{proof}
Similarly to the case of cycles, if one could shave the $\log(N)$ factor off in the case of the hemicube and prove that it displays the same soundness as the standard cube, we would immediately obtain a $1/\mathrm{polylog}(N)$ soundness for the generalized hemicube code.

\section{Efficient decoding algorithm for adversarial errors}

\label{sec:decoding}

In this section, we explain how the small fillings and cofillings promised by the results of the previous section can be exploited to give an efficient decoding algorithm with good performance against adversarial errors. 
The main idea is to notice that one can efficiently find such fillings and cofillings and therefore find Pauli errors giving the observed syndrome. While finding the smallest possible fillings or cofillings does not appear to be easy, finding ones satisfying the bounds of Lemmas \ref{thm:chain} and \ref{thm:cochain} can be done efficiently.

We note, however, that the decoding algorithm does not seem to perform so well
against random errors of linear weight. In particular, any argument based on
percolation theory that would say that errors tend to only form small clusters
and that therefore it is sufficient to correct these errors (similarly to
\cite{FGL18} for instance) fail here because of the logarithmic weight of the
generators. Indeed, the factor graph of the code has logarithmic degree and
there does not exist a constant threshold for the error probability such that
below this threshold, errors appear in clusters of size $o(N)$. Nevertheless,
it seems that a decoding algorithm such as the small set flip algorithm of
\cite{LTZ15} performs rather well for the hemicubic code. 

For simplicity, we restrict our attention to the single-qubit code in the following.

\decoding*

The decoding complexity is quasilinear in the error size and can be done in logarithmic depth.\\

We first review the complexity of finding a small filling in the Hamming cube (without identifying antipodal faces) using the construction of Lemma \ref{lem:dott1}. 
Starting with a $(p-1)$-cycle $Z = 0 Z_0 + * Z_* +1 Z_1$, one picks a random cut and recursively defines the corresponding filling
\begin{align*}
Y = * Z_1 + 0 D_{n-1, p}(Z_0+Z_1).
\end{align*}
Exploiting Lemma \ref{lem:dott1}, one can bound the size of $Y$ as follows:
\begin{align}
\|Y\| \leq \frac{n-p}{2p} \|Z_0\| + \frac{n+p}{2p} \|Z_1\|. \label{eqn:dott-algo}
\end{align}
Choosing the cut which minimizes the right hand size can be done efficiently as
it simply amounts to computing $\|Z_0\|$ and $\|Z_*\|$ for the $n$ possible
cuts, which has complexity $n \|Z\|$. By choosing the optimal cut, one
guarantees that the filling $Y$ satisfies the bound $\|Y\| \leq
\frac{n-p+1}{2p}\|Z\|$. (It is not even needed to find the optimal cut,
since any cut such that $\frac{n-p}{2p} \|Z_0\| + \frac{n+p}{2p} \|Z_1\| \leq \frac{n-p+1}{2p}\|Z\|$ yields a filling satisfying the final bound.)
This gives an algorithm of complexity $O(n^2 \|Z\|)$. Finding a cofilling in the Hamming cube can be done similarly by exploiting Lemma \ref{lem:dott2}. 

Recall that as usual, it is sufficient to correct for Pauli errors since they
form a basis of all possible errors. Moreover, we can choose to correct
$X$-errors and $Z$-errors independently. In the case of the hemicubic code, it means that we are given two syndromes corresponding to a boundary and a coboundary, and that we should find a filling and a cofilling of these two syndromes. 
This is done by applying the algorithms of Lemmas \ref{thm:chain} and \ref{thm:cochain}. 
For instance, finding a Pauli-$X$ error giving the correct syndrome $\partial E_X = Z$ amounts to choosing a symmetric filling as follows:
\begin{align}
Y = * \tilde{D}(Z_*) + 0 D (Z_0 + \tilde{D}(Z_*) ) + 1 \overline{ D (Z_0 + \tilde{D}(Z_*) )},
\end{align}
where $D$ is the (not necessarily symmetric) filling promised by Lemma \ref{lem:dott1} and $\tilde{D}$ is the symmetric filling defined in Lemma \ref{thm:chain}.
Like before, one can bound the size of this filling:
\begin{align}
\|Y \|  \leq \frac{n-p}{p} \|Z_0\| + \frac{n}{p} c_{n-1,p-1} \|Z_*\| 
\end{align}
where
\begin{align} \label{eqn:dott-algo-symm}
c_{n,p}= \frac{(n-p+1)(n-p)}{2p} \sum_{m=n-p+1}^{n} \frac{1}{m}.
\end{align}
Again, if we are not interested in the smallest filling, but simply one satisfying the promised bound, it is possible to find it efficiently by choosing a cut minimizing the rhs of Eq.~\eqref{eqn:dott-algo-symm}. This again has complexity $O(n \|Z\|)$ at a given level. 

Overall, finding a small symmetric filling has complexity $O(n^4 \|Z\|)$, where we recall that $n$ is logarithmic in the length of the quantum code. 
Correcting for $Z$ errors is done similarly using the algorithm for cofillings (Lemma \ref{thm:cochain}) instead.

Let us now show that this algorithm recovers the correct error (up to a stabilizer element), and therefore that decoding succeeds.
Let $Y$ be the support of a Pauli-$X$ error and denote by $Z = \partial Y$ its syndrome. Note that $\|Z\| \leq 2p+1 \|Y\|$ since the generators have weight $2p$.  The algorithm described above yields a chain $Y'$ such that $\partial Y' = Y = \partial Y$ and of size $\|Y'\| \leq c_{n,p} \|Z\| \leq 2p c_{n,p} \|Y\|$
Observe now that the following inequalities hold:
\begin{align*}
\|Y + Y'\| & \leq \|Y\| + \|Y'\| \leq (1+ 2p c_{n,p}) \|Y\| \leq (1+2p(n-p)) \|Y\|.
\end{align*}
In particular, as long as $(1+2p(n-p)) \|Y\| < \tbinom{n}{p}$, the cycle $Y+Y'$ cannot yield a logical error and the decoding was successful.
Similarly, if $Y$ is the support of a Pauli-$Z$ error, the same reasoning shows that the decoding is successful as long as $(1+2(n-p) c'_{n,p}) \|Y\| < 2^{n-p-1}$. 

Combining both conditions, we obtain that the decoding is successful for any error of weight less than $\frac{d_{\min}}{2p(n-p)+1}$.


\end{document}